\documentclass[11pt]{article}

\usepackage[margin=1in,footskip=0.25in]{geometry}

\date{\today}

\newif\ifsubmission
\submissiontrue

\usepackage[sort&compress,square,comma,numbers]{natbib}

\usepackage[pagebackref]{hyperref}
\hypersetup{
    colorlinks=true,      linkcolor=PineGreen, 
citecolor=NavyBlue,      filecolor=Purple,    urlcolor=Purple         }

\newif\ifcolt
\coltfalse

\usepackage[utf8]{inputenc}
\usepackage[T1]{fontenc}

\let\citet\cite
\let\citep\cite

\usepackage{fourier}
\DeclareMathAlphabet{\mathcal}{OMS}{cmsy}{m}{n}

\usepackage{amsmath,amsthm}
\usepackage{aliascnt}

\usepackage{float}
\usepackage{graphicx}

\usepackage[capitalize, nameinlink, noabbrev]{cleveref}

\newtheorem{theorem}{Theorem}
\numberwithin{theorem}{section}

\newcommand{\makealiasedtheorem}[3]{\newaliascnt{#1}{theorem}\newtheorem{#1}[#1]{#2}\aliascntresetthe{#1}\crefname{#1}{#2}{#3}\Crefname{#1}{#2}{#3}}

\theoremstyle{plain}
\makealiasedtheorem{prop}{Proposition}{Propositions}
\makealiasedtheorem{conjecture}{Conjecture}{Conjectures}
\makealiasedtheorem{corollary}{Corollary}{Corollaries}
\makealiasedtheorem{lem}{Lemma}{Lemmas}
\makealiasedtheorem{example}{Example}{Examples}
\makealiasedtheorem{fact}{Fact}{Facts}
\makealiasedtheorem{claim}{Claim}{Claims}
\makealiasedtheorem{assumption}{Assumption}{Assumptions}
\makealiasedtheorem{setting}{Setting}{Settings}
\makealiasedtheorem{observation}{Observation}{Observations}
\makealiasedtheorem{problem}{Problem}{Problems}
\makealiasedtheorem{importedtheorem}{Imported Theorem}{Imported Theorems}
\makealiasedtheorem{inftheorem}{Informal Theorem}{Informal Theorems}
\makealiasedtheorem{importedlemma}{Imported Lemma}{Imported Lemmas}
\makealiasedtheorem{incompletetheorem}{Incomplete Theorem}{Incomplete Theorems}

\theoremstyle{definition}
\makealiasedtheorem{construction}{Construction}{Constructions}
\makealiasedtheorem{definition}{Definition}{Definitions}

\theoremstyle{remark}
\makealiasedtheorem{remark}{Remark}{Remarks}

\newenvironment{cor}
    {\begin{corollary}
    }
    {
    \end{corollary}
    }

\numberwithin{equation}{section}

\crefname{equation}{}{}
\Crefname{equation}{Equation}{Equations}
\crefname{item}{item}{items}
\Crefname{item}{Item}{Items}

\usepackage{amsmath}
\usepackage{amssymb}

\usepackage{algorithm, algpseudocode, algorithmicx}
\usepackage{bm}
\usepackage{bbm}
\usepackage{booktabs}
\usepackage{cancel}
\usepackage{caption}
\usepackage{commath}
\usepackage{enumitem}
\usepackage{mathtools}
\usepackage{nicefrac}       \usepackage{relsize} \usepackage{subcaption}
\usepackage[dvipsnames,table]{xcolor}
\usepackage{xspace}
\usepackage{url}
\usepackage{breakcites}
\usepackage{stmaryrd}

\usepackage{titlesec}
\titlespacing*{\section}{0pt}{12pt}{5pt}
\titlespacing*{\subsection}{0pt}{11pt}{5pt}
\titlespacing*{\subsubsection}{0pt}{11pt}{5pt}
\titlespacing*{\paragraph}{0pt}{6pt}{1em}

\floatstyle{plaintop}
\restylefloat{table}

\usepackage{crossreftools}
\renewcommand{\algorithmicrequire}{\textbf{Input:}}
\renewcommand{\algorithmicensure}{\textbf{Output:}}

\algrenewcommand\alglinenumber[1]{\sf\scriptsize\color{NavyBlue}{#1}}
\algrenewcommand\algorithmicrequire{\textbf{Input:}}
\algrenewcommand\algorithmicensure{\textbf{Output:}}

\usepackage{contour}
\usepackage[normalem]{ulem}

\contourlength{0.8pt}

\makeatletter
\def\hlinewd#1{\noalign{\ifnum0=`}\fi\hrule \@height #1 \futurelet
	\reserved@a\@xhline}
\makeatother

\renewcommand*{\backref}[1]{}
\renewcommand*{\backrefalt}[4]{\ifcase #1 (No citations.)\or
  (Cited on page #2.)\else
  (Cited on pages #2.)\fi
}

\makeatletter
\newtheorem*{rep@theorem}{\rep@title}
\newcommand{\newreptheorem}[2]{\newenvironment{rep#1}[1]{\def\rep@title{\Cref{##1}, restated}\begin{rep@theorem}}{\end{rep@theorem}}}
\makeatother
\makeatletter
\newtheorem*{rep@lemma}{\rep@title}
\newcommand{\newreplemma}[2]{\newenvironment{rep#1}[1]{\def\rep@title{\Cref{##1} Restated}\begin{rep@lemma}}{\end{rep@lemma}}}
\makeatother
\makeatletter
\newtheorem*{rep@definition}{\rep@title}
\newcommand{\newrepdefinition}[2]{\newenvironment{rep#1}[1]{\def\rep@title{\Cref{##1}, restated}\begin{rep@definition}}{\end{rep@definition}}}
\makeatother
\makeatletter
\newtheorem*{rep@corollary}{\rep@title}
\newcommand{\newrepcorollary}[2]{\newenvironment{rep#1}[1]{\def\rep@title{\Cref{##1}, restated}\begin{rep@corollary}}{\end{rep@corollary}}}
\makeatother
\newreptheorem{theorem}{Theorem}
\newreplemma{lemma}{Lemma}
\newrepdefinition{definition}{Definition}
\newrepcorollary{corollary}{Corollary}

\usepackage[most]{tcolorbox}
\newcommand{\actionbox}[1]{\begin{tcolorbox}[colback=white,colframe=black,width=\columnwidth,boxsep=5pt,arc=4pt]
    \centering \emph{#1}
\end{tcolorbox}}

\newcommand{\e}{\mathrm{e}}
\DeclareMathOperator{\TV}{\mathrm{D}_{\mathrm{TV}}}

\renewcommand{\top}{\protect{\smash{\intercal}}}

\let\daggerFake\dagger
\renewcommand{\dagger}{{\smash{\daggerFake}}}

\newcommand{\defeq}[0]{\ensuremath{\;{\vcentcolon=}\;}\xspace}

\definecolor{ltyellow}{rgb}{1, 1, 0.9}
\usepackage{soul}
\sethlcolor{ltyellow}
\newcommand{\warn}[1]{\emph{#1}}

\DeclareMathOperator{\TriBlockKrylov}{TriBlkKrylov}
\DeclareMathOperator{\BlockKrylov}{BlkKrylov}
\DeclareMathOperator{\BilinearPower}{BilinearPower}
\DeclareMathOperator{\Wishart}{\textsc{Wishart}}

\DeclareMathOperator{\Haar}{\textsc{haar}}

\DeclareMathOperator{\polylog}{polylog}

\let\norm\relax
\DeclarePairedDelimiter\norm{\lVert}{\rVert}

\DeclareMathOperator*{\Var}{Var}
\DeclareMathOperator*{\E}{\mathbb{E}}
\DeclareMathOperator{\tr}{tr}
\DeclareMathOperator{\diag}{diag}

\newcommand{\etal}{\text{et al.}\xspace}

\newcommand{\eps}[0]{\ensuremath{\varepsilon}}
\let\epsilon\eps

 \newcommand{\vecalt}[1]{\boldsymbol{#1}}

\newcommand{\mat}[1]{\bm{#1}} \renewcommand{\vec}[1]{\bm{#1}}

 \newcommand{\bmat}[1]{\begin{bmatrix} #1 \end{bmatrix}} \newcommand{\flatbmat}[1]{[\begin{matrix} #1 \end{matrix}]}  \newcommand{\sbmat}[1]{\left[\begin{smallmatrix} #1 \end{smallmatrix}\right]}

\newcommand{\coloneqq}{\defeq}
\newcommand{\prob}{\mathbb{P}}
\renewcommand{\Pr}{\prob}
\newcommand{\Id}{\mathbf{I}}
\newcommand{\order}{\cO}
\newcommand{\orderish}{\widetilde\cO}

\definecolor{color1}{HTML}{2437E6}
\definecolor{color2}{HTML}{D12757}

\newcommand{\mA}{\ensuremath{\mat{A}}\xspace}

\newcommand{\mD}{\ensuremath{\mat{D}}\xspace}

\newcommand{\mG}{\ensuremath{\mat{G}}\xspace}

\newcommand{\mI}{\ensuremath{\mathbf{I}}\xspace}

\newcommand{\mM}{\ensuremath{\mat{M}}\xspace}

\newcommand{\mQ}{\ensuremath{\mat{Q}}\xspace}

\newcommand{\mT}{\ensuremath{\mat{T}}\xspace}
\newcommand{\mU}{\ensuremath{\mat{U}}\xspace}
\newcommand{\mV}{\ensuremath{\mat{V}}\xspace}
\newcommand{\mW}{\ensuremath{\mat{W}}\xspace}

\newcommand{\mY}{\ensuremath{\mat{Y}}\xspace}
\newcommand{\mZ}{\ensuremath{\mat{Z}}\xspace}

\renewcommand{\eqref}{\cref}

\let\oldthebibliography\thebibliography
\renewcommand{\thebibliography}[1]{\oldthebibliography{#1}\setlength{\itemsep}{0pt}\setlength{\parskip}{0pt}}

\newcommand{\vb}{\ensuremath{\vec{b}}\xspace}
\newcommand{\vc}{\ensuremath{\vec{c}}\xspace}

\newcommand{\ve}{\ensuremath{\mathbf{e}}\xspace}

\newcommand{\vg}{\ensuremath{\vec{g}}\xspace}

\newcommand{\vr}{\ensuremath{\vec{r}}\xspace}

\newcommand{\vu}{\ensuremath{\vec{u}}\xspace}
\newcommand{\vv}{\ensuremath{\vec{v}}\xspace}

\newcommand{\vx}{\ensuremath{\vec{x}}\xspace}
\newcommand{\vy}{\ensuremath{\vec{y}}\xspace}
\newcommand{\vz}{\ensuremath{\vec{z}}\xspace}

\newcommand{\vzero}{\ensuremath{\vec{0}}\xspace}

\newcommand{\vdelta}{\ensuremath{\vecalt{\delta}}\xspace}

\newcommand{\cA}{\ensuremath{{\mathcal A}}\xspace}
\newcommand{\cB}{\ensuremath{{\mathcal B}}\xspace}

\newcommand{\cE}{\ensuremath{{\mathcal E}}\xspace}
\newcommand{\cF}{\ensuremath{{\mathcal F}}\xspace}

\newcommand{\cM}{\ensuremath{{\mathcal M}}\xspace}
\newcommand{\cN}{\ensuremath{{\mathcal N}}\xspace}
\newcommand{\cO}{\ensuremath{{\mathcal O}}\xspace}

\newcommand{\cS}{\ensuremath{{\mathcal S}}\xspace}
\newcommand{\cT}{\ensuremath{{\mathcal T}}\xspace}

\newcommand{\bbC}{\ensuremath{{\mathbb C}}\xspace}

\newcommand{\bbN}{\ensuremath{{\mathbb N}}\xspace}
\newcommand{\bbO}{\ensuremath{{\mathbb O}}\xspace}

\newcommand{\bbR}{\ensuremath{{\mathbb R}}\xspace}

\newcommand{\rC}{\ensuremath{\mathrm{C}}\xspace}

\DeclareMathOperator{\cond}{cond}

\title{The matrix-vector complexity of $\mA\vx=\vb$}

\author{
    Michał Dereziński \\ University of Michigan \\ \texttt{derezin@umich.edu}
    \and
    Ethan N. Epperly \\ UC Berkeley \\ \texttt{eepperly@berkeley.edu}
    \and
    Raphael A. Meyer \\ UC Berkeley \& ICSI \\ \texttt{ram900@berkeley.edu}
}

\begin{document}
\maketitle

\begin{abstract}
    Matrix--vector algorithms, particularly Krylov subspace methods, are widely viewed as the most effective algorithms for solving large systems of linear equations.
    This paper establishes lower bounds on the worst-case number of matrix--vector products needed by such an algorithm to approximately solve a general linear system.
    The first main result is that, for any matrix--vector algorithm which is allowed the use of randomization and can perform products with both a matrix and its transpose, $\Omega(\kappa \log(1/\varepsilon))$ matrix--vector products are necessary to solve a linear system with condition number $\kappa$ to accuracy $\varepsilon$, matching an upper bound for conjugate gradient on the normal equations.
    The second main result is that one-sided algorithms, which lack access to the transpose, must use $n$ matrix--vector products to solve an $n \times n$ linear system, even when the problem is perfectly conditioned.
    Both main results include explicit constants that match known upper bounds up to a factor of four.
    These results rigorously demonstrate the limitations of matrix--vector algorithms and confirm the optimality of widely used Krylov subspace algorithms.
\end{abstract}

\section{Introduction}

Solving systems of linear algebraic equations is a fundamental problem in computer science and mathematics, and algorithms for this task are used in almost every area of modern computation, including machine learning \citep{mohri2018foundations}, scientific computing \citep{greenbaum_iterative_1997}, and optimization \citep{boyd2004convex}.
Given an invertible matrix $\mA \in \bbR^{n \times n}$ and a vector $\vb \in \bbR^n$, the objective is to find the vector $\vx \in \bbR^n$ for which $\mA \vx = \vb$.
The standard algorithms for this task are variants of Gaussian elimination, which run in $\order(n^\omega)$ operations when implemented using fast matrix multiplication.\footnote{Here, $\omega \le 2.371\ldots$ is the matrix multiplication exponent.}
Determining the optimal complexity for the linear system problem is a major open question \citep{spielman_solve_2024}.

If we are willing to settle for an \warn{approximate} solution to $\mA\vx = \vb$, the design space for algorithms becomes larger.
Significant attention has gone to the class of matrix--vector algorithms, which learn about \mA \emph{only} through the matrix--vector product (matvec) primitives $\vz \mapsto \mA\vz$ and $\vz \mapsto \mA^\top \vz$.
These methods are commonly understood to be among the only approaches for solving large linear systems \citep{greenbaum_iterative_1997}.
The most popular matrix--vector algorithms for solving linear systems are \emph{Krylov subspace methods}, such as the conjugate gradient and GMRES algorithms, which see wide use in practice.
Given the importance of these methods, it is natural to ask:
\actionbox{
How many matvecs are neccessary and sufficient to approximately solve \(\mA\vx=\vb\)?
}
\noindent 
To specify this problem completely, we must determine an accuracy requirement for an approximate solution \(\tilde\vx\).
For this article, we seek a vector $\tilde\vx$ with a small relative residual:
\begin{equation} \label{eq:approximate-solution}
    \norm{{\mA\tilde\vx - \vb}}_2 \le \varepsilon \norm{\vb}_2 \quad \text{for specified } \varepsilon > 0.
\end{equation}
Throughout this paper, $\norm{\cdot}_2$ will denote the $\ell_2$ norm of a vector or the spectral norm of a matrix.
See \cref{app:other-metrics} for a discussion of other error metrics.

\subsection{Background and research questions} \label{sec:background}

Most work on matrix--vector complexity for the $\mA\vx = \vb$ problem focuses on the case when \mA is symmetric positive definite (SPD).
In this case, the standard Krylov subspace algorithms are conjugate gradient and MINRES \cite[Algs.~2 \& 4]{greenbaum_iterative_1997}.
The latter achieves the guarantee \cref{eq:approximate-solution} using $\order(\sqrt{\kappa}\, \log(\nicefrac1\eps))$ matvecs.
Here, $\kappa \coloneqq \cond(\mA) \coloneqq \norm{\mA}_2 \norm{\smash{\mA^{-1}}}_2 = \sigma_{\mathrm{max}}(\mA) / \sigma_{\mathrm{min}}(\mA)$ is the \smash{\emph{condition number}}.
Classical lower bounds confirm that \(\Omega(\sqrt{\kappa}\,\log(\nicefrac1\eps))\) matvecs are necessary for any deterministic algorithm to achieve this guarantee \citep[Sec.~7.2]{nemirovskij1983problem}. However, with the rapidly increasing use of randomization in computational linear algebra \citep{martinsson2020randomized,randlapack_book,derezinski2024recent}, a lower bound against deterministic algorithms is insufficient to understand the difficulty of solving a given problem.
This limitation was partially addressed by
\ifcolt 
\citet{BHSW20},
\else
Braverman \etal \citep{BHSW20},
\fi 
who proved a weaker \(\Omega(\sqrt\kappa/\polylog(\kappa))\) lower bound for SPD linear systems against randomized~algorithms.

It is natural to inquire how the matrix--vector complexity (of both deterministic and randomized algorithms) changes when we consider general linear systems for which \mA is not SPD.
In this setting, we can distinguish between a \emph{two-sided matrix--vector algorithm}, which can perform matrix--vector products with both \mA and $\mA^\top$, and a \emph{one-sided matrix--vector algorithm}, which can only perform matrix--vector products with \mA.
Both one- and two-sided algorithms are studied and used in practice \citep{greenbaum_iterative_1997}. 

To solve $\mA\vx = \vb$ using a two-sided matrix--vector algorithm, we can convert the problem into an SPD problem by passing to the \emph{normal equations}
\begin{equation*}
    (\mA^\top \mA)\vx = \mA^\top \vb.
\end{equation*}
The matrix $\mA^\top \mA$ is SPD,
so the conjugate gradient method can solve this system.
The condition number of the normal equations is larger, \(\cond(\mA^\top\mA) = \kappa^2\), implying that conjugate gradient on the normal equations
requires $\order(\sqrt{\kappa^2} \kern.05em \log(\nicefrac1\eps)) = \order(\kappa \log(\nicefrac1\eps))$ matvecs to produce an $\varepsilon$-accurate solution \citep[Eqn.~2.12]{nachtigal_how_1992}.
This is known to be tight for deterministic algorithms \citep[p.~180]{traub_information_1988}. 
For randomized algorithms, no existing lower bounds extend beyond the SPD case.
We ask:
\actionbox{
Does \emph{every} two-sided algorithm require $\Omega(\kappa\log(\nicefrac1\eps))$ matvecs in the worst case?
}

The most popular one-sided matrix--vector algorithm is the GMRES algorithm \citep{saad_gmres_1986}.
This algorithm has many desirable properties: When it is applied to an SPD system, it achieves the same $\order(\sqrt{\kappa}\kern.08em  \log(\nicefrac1\eps))$ complexity as MINRES, and it converges rapidly on many problems in practice.
However, this algorithm also has a fundamental weakness: There exists a perfectly conditioned linear system ($\kappa=1$) on which GMRES requires a full $n$ iterations to solve the system to any accuracy level $\varepsilon < 1$ \cite[Ex.~C]{nachtigal_how_1992}.
This raises the question:
\actionbox{
Does \emph{every} one-sided algorithm require $\Omega(n)$ matvecs in the worst case, even when $\kappa\!=\!1$?
}

\subsection{Our results} \label{sec:results}

This paper answers both of these questions in the affirmative.
Our first result concerns two-sided matrix--vector algorithms.

\begin{theorem}[Linear systems: Lower bound against two-sided algorithms]
    \label{thm:intro-two-sided-lower-bound}
    Fix \(\eta>0\).
    There does not exist any randomized algorithm that takes inputs \(\kappa\), \(\eps\), and \vb, which computes fewer than \(\frac{1-\eta}{4}\kappa\log(\nicefrac1\eps)\) two-sided matrix--vector products with \mA, and which returns a vector \(\tilde\vx\) such that
\begin{equation*}
        \norm{\mA\tilde\vx - \vb}_2 \le \varepsilon \norm{\vb}_2 \quad \text{with prob.\ } \ge \frac{5}{6}
    \end{equation*}
for all numbers \(\eps>0\) and \(\kappa \geq 1\) and matrices \mA of condition number $\cond(\mA)\le \kappa$. 
\end{theorem}

This result shows that conjugate gradient on the normal equations is optimal from the perspective of worst-case complexity.\footnote{We note that there are reasons in practice to avoid iterative methods based on the normal equations, including numerical stability \cite[Sec.~8.2]{paige_lsqr_1982} and better-than-worst case performance \citep{nachtigal_how_1992}.}
A strength of this result is its explicitness, with the complexity bound featuring no hidden constants.
Conjugate gradient on the normal equations achieves the accuracy guarantee \cref{eq:approximate-solution} in \((1+o(1))\kappa \log(\nicefrac1\eps)\) matvecs \cite[Eqn.~3.6]{greenbaum_iterative_1997}, so the lower and upper bounds match to a factor of roughly 4.

Our second result concerns one-sided, transpose-free algorithms.

\begin{theorem}[Linear systems: Lower bound against one-sided algorithms]
    \label{thm:intro-one-sided-lower-bound}
    Consider any randomized algorithm that takes a vector \vb and a one-sided matrix--vector oracle for a matrix \mA and outputs a vector \(\tilde\vx\), and suppose $n$ is larger than some universal constant.
Then: \begin{align*}
        \text{Achieving }\ \norm{\mA\tilde\vx - \vb}_2 &\le \frac{1}{2} \norm{\vb}_2 &&\text{with prob.\ } \ge \frac{2}{3} &&\text{requires at least $\lceil n/2\rceil$ matvecs}; \\
        \text{Achieving }\ \norm{\mA\tilde\vx - \vb}_2 &\le \frac{0.2}{\sqrt{n}} \norm{\vb}_2 &&\text{with prob.\ } \ge \frac{2}{3} &&\text{requires $n$ matvecs}.
    \end{align*}
Moreover, the hard instance $\mA$ can be taken to be \emph{orthogonal} ($\cond(\mA) = 1$).
\end{theorem}

This result confirms that GMRES's requirement of all $n$ matvecs is optimal, even for perfectly conditioned systems.
More precisely, to solve $\mA\vx = \vb$ to constant accuracy requires accessing half the matrix ($n/2$ matvecs), and solving to accuracy $\order(1/\sqrt{n})$ requires reading \warn{the entire} matrix.
Unless we assume some structure on \mA \emph{beyond its condition number}, access to the transpose \(\mA^\top\) is \emph{necessary} to solve linear systems of equations with better than trivial complexity.

Algorithms like conjugate gradient and GMRES are already known to be optimal among certain classes of algorithms, like within the class of deterministic single-vector Krylov methods (see
\ifcolt
\citet[Ch.~2]{greenbaum_iterative_1997}
\else 
\cite[Ch.~2]{greenbaum_iterative_1997}
\fi and
\ifcolt
\citet[Sec.~7.2]{nemirovskij1983problem}).
\else 
\cite[Sec.~7.2]{nemirovskij1983problem}).
\fi
However, prior works do not fully answer if these methods are optimal within the larger class of general randomized matrix--vector algorithms.\footnote{
\ifcolt
\cite{BHSW20}
\else 
Braverman \etal \citep{BHSW20}
\fi 
prove that \(\Omega(\sqrt\kappa / \polylog(\kappa))\) matvecs are needed to solve a SPD linear system to constant error, partially resolving this question.}
Indeed, there have been many examples of randomized algorithms that achieve surprising improvements to the complexity of problems in linear algebra \citep{rokhlin_fast_2008,clarkson_numerical_2009,clarkson_low-rank_2017,musco_sublinear_2017,meyer_hutch_2021,peng_solving_2024,DerezinskiSidfordSODA26}.
However, despite all efforts, no improvements to the worst-case complexity of solving general linear systems have materialized.
Our research demonstrates why this is the case: Linear system solving is \emph{not} amenable to such a speedup, in both the one- and two-sided cases.

\paragraph{Our Techniques.~}
We use different techniques to handle the one- and two-sided cases.
The two-sided lower bound (\cref{thm:intro-two-sided-lower-bound}) uses an approach based on the \emph{nonexistence} of good polynomial approximations, extending an argument of
\ifcolt
\citet{chewi2024query}.
\else 
Chewi \etal \cite{chewi2024query}.
\fi
We show that since no polynomial can approximate the function \(f(x)=\nicefrac1x\) to error \(\eps\) on the split interval \(\cS \defeq [-\kappa,-1]\cup[1,\kappa]\) with degree less than \(t=\cO(\kappa\log(1/\eps))\), no algorithm can use fewer than \(\cO(t)\) matvecs to estimate \(\tr(f(\mA)) = \tr(\mA^{-1})\) for all matrices with eigenvalues in \cS.
We then make a reduction, showing that any linear system solver satisfying \cref{eq:approximate-solution} can be used to estimate \(\tr(\mA^{-1})\), resulting in a proof of \cref{thm:intro-two-sided-lower-bound}.
See \cref{sec:two-sided}~for~details.

The one-sided lower bound (\cref{thm:intro-one-sided-lower-bound}) follows from a new \emph{hidden Haar theorem} (\cref{thm:hidden-haar}), which shows that a one-sided matrix--vector algorithm cannot learn much information about the rows of a matrix unless it expends \(\Omega(n)\) matvecs.
\Cref{thm:intro-one-sided-lower-bound} follows because, for an orthogonal linear system $\mQ\vx = \vb$, computing the solution \(\mQ^{-1}\vb = \mQ^\top\vb\) requires learning a linear combination of \mQ's rows.
See \cref{sec:one-sided} for details.

\subsection{Extension: Stronger lower bounds for solving SPD systems}

The focus of this article is on solving general, not necessarily SPD, linear systems.
Indeed, the prior work of Braverman \etal \cite{BHSW20} has already shown that \(\Omega(\sqrt\kappa/\polylog(\kappa))\) matvecs are needed for a randomized algorithm to solve a SPD linear system.
However, our analysis is able to provide tighter lower bounds for the SPD case as well, removing unnecessary polylogarithmic factors in the condition number $\kappa$ and obtain the correct logarithmic dependence on the inverse-accuracy $\nicefrac1\eps$:
\begin{theorem}[SPD Linear systems: Lower bound against two-sided algorithms]
    \label{thm:intro-spd-lower-bound}
    Fix \(\eta>0\).
    There does not exist any randomized algorithm that takes inputs \(\kappa\), \(\eps\), and \vb, which computes fewer than \(\frac{1-\eta}{4} \sqrt\kappa\log(\nicefrac1\eps)\) two-sided matrix-vector products with \mA, and which returns a vector \(\tilde\vx\) such that
    \[
        \norm{\mA\tilde\vx-\vb}_2 \leq \eps\norm{\vb}_2
        \qquad
        \text{with prob. } \geq \frac56
    \]
    for all numbers \(\eps > 0\) and \(\kappa \geq 1\) and SPD matrices \mA of condition number \(\cond(\mA) \leq \kappa\).
\end{theorem}
This result sharpens the log factors over the prior work and, since conjugate gradient converges in \(\cO(\sqrt\kappa\log(\nicefrac1\eps))\) matvecs, shows that the dependence on \(\eps\) in the upper bound analysis is tight.
The proof of \cref{thm:intro-spd-lower-bound} uses the same basic argument as the proof of the lower bound in the general case \cref{thm:intro-two-sided-lower-bound}.
See \cref{sec:lin-sys-lower-bound-psd} for details.

\subsection{Implication: Fine-grained lower bounds and complexity separation}

Our main results describe the complexity of solving linear systems by parameterizing the input via the matrix dimension $n$ and its condition number $\kappa = \cond(\mA)$. Yet, it is widely accepted that Krylov methods often dramatically outperform the worst-case complexity bounds when the singular values of $\mA$ exhibit additional structure, such as clustering or outliers \citep{greenbaum_iterative_1997}.

An important and well-studied example of this arises when $\mA$ has a number of large outlying singular values that blow up its condition number, which arises in smoothed analysis \citep{spielman2009smoothed}, optimization \citep{boyd2004convex}, and machine learning \citep{zhang2013divide}. 
In this setting, conjugate gradient on the normal equations is known to solve such linear systems in an improved \(\cO(k + \kappa_k\log(1/\eps))\) matvecs \citep{axelsson1986rate},
where $\kappa_k\coloneqq\sigma_{k+1}(\mA)/\sigma_{\mathrm{min}}(\mA)$ denotes the condition number without the top \(k\) singular values.
This more fine-grained notation of complexity, parameterized by both \(k\) and \(\kappa_k\), circumvents our basic two-sided lower bound (\cref{thm:intro-two-sided-lower-bound}).
However, our results can be extended to this different complexity measure, again showing that conjugate gradient on the normal equations is optimal:

\begin{theorem}[Linear systems: Fine-grained lower bound]\label{c:fine-grained}
    There does not exist a randomized algorithm that takes inputs $k$, $\kappa_k$, $\varepsilon$, and $\vb$, which computes fewer than $\Omega(k+ \kappa_k\log(1/\epsilon))$ two-sided matrix-vector products with $\mA$ and returns $\tilde\vx$ such~that
    \begin{align*}
        \|\mA\tilde\vx-\vb\|_2\leq \varepsilon\|\vb\|_2\quad\text{with prob.\ }\geq \frac 56
    \end{align*}
    for all matrices $\mA$ such that $\sigma_{k+1}(\mA) / \sigma_{\mathrm{min}}(\mA) \leq \kappa_k$.
\end{theorem}
This corollary follows immediately by combining our results with a reduction from
\ifcolt
\citet[Lem.~27 \& Rem.~8]{derezinski2025fine}
\else
Dereziński \etal \cite[Lem.~27 \& Rem.~8]{derezinski2025fine}
\fi
; see \cref{app:fine-grained} for details.

Notably, this result implies a \emph{complexity separation} between matrix--vector algorithms and general algorithms that can access the entries of $\mA$ directly.
If \mA is a dense matrix and we implement matvecs using the standard algorithm, then \cref{c:fine-grained} shows that \(\Omega(n^2k + n^2\kappa_k\log(1/\eps))\) time is needed to solve a linear system.
However, recent work \citep{derezinski2024solving,derezinski_faster_2025,DerezinskiSidfordSODA26} shows that if \mA is stored as a dense array, we can solve $\mA\vx=\vb$ in $\orderish(k^3 + n^2\kappa_k\log(1/\epsilon))$ time.
So, for large $n$ and $k$, having access to the matrix offers a direct improvement over matrix--vector access, demonstrating a complexity separation between matrix--vector algorithms and general algorithms.
This is the first result of this kind for general, dense~linear~systems.

\subsection{Related work}

The most relevant background material, and its relationship to our work, were discussed in \cref{sec:background,sec:results}.
We briefly mention a few more areas of related work.

\paragraph{Upper and lower bounds for Krylov methods.}
Understanding the behavior of Krylov linear solvers has been a decades-long research challenge in numerical analysis.
The literature is extensive; see
\ifcolt
\cite{greenbaum_iterative_1997}
\else 
Greenbaum's book \citep{greenbaum_iterative_1997}
\fi 
for a summary of the theory.
Lower bounds for the rate of convergence of Krylov methods can be derived from potential theory \citep{driscoll_potential_1998}.
These lower bounds apply only to specific Krylov methods, not general randomized matrix--vector algorithms.

\paragraph{Transpose-free linear algebra.} Our work fits into the growing research topic of \emph{transpose-free linear algebra} \citep{boulle_operator_2024,halikias_structured_2025}, which aims to understand the relative power of one- and two-sided algorithms for various linear algebra problems.
The results most similar to ours appear in the thesis of 
\ifcolt 
\citet[Sec.~5.2]{halikias_structured_2025},
\else 
Halikias \cite[Sec.~5.2]{halikias_structured_2025},
\fi 
who analyzes one-sided matrix--vector algorithms for solving linear systems and least-squares problems.
Her results show that, even in the favorable situation when $\cond(\mA) = 1$ and we perform $n-1$ matvecs $\mA\vz_1,\ldots,\mA\vz_{n-1}$,\footnote{Halikias also includes as an assumption that the exact solution $\mA^{-1}\vb$ is not in the span of the queries $\{\vz_1,\ldots,\vz_{n-1}\}$.} the data collected by the algorithm is not sufficient to characterize the solution to a system of linear equations:
There exists another well-conditioned matrix $\tilde\mA$ that is undistinguished by the queries
\begin{equation*}
    \mA\vz_i = \tilde\mA\vz_i \quad \text{for } i=1,\ldots,n-1
\end{equation*}
but for which the solutions of the linear systems are different
\begin{equation*}
    \norm{\tilde\mA^{-1}\vb - \mA^{-1}\vb}_2 > \delta \quad \text{for some } \delta > 0 \text{ \warn{depending on $\mA,\vb,\vz_1,\ldots,\vz_{n-1}$}}.
\end{equation*}
Halikias's result establishes information-theoretic limitations on the power of one-sided matrix--vector algorithms.
However, the level of accuracy $\delta$ and the hard-to-distinguish instance $\tilde\mA$ \warn{depend on the random choices made by the algorithm}.
To establish lower bounds on matrix--vector complexity, we must exhibit a problem instance $(\mA,\vb)$ and a fixed level of accuracy $\delta > 0$, independent of the random choices made by the algorithm, for which any algorithm fails.
As such, while our analysis and Halikias' share similar aims, our results are formally incomparable.

\paragraph{Quantum linear systems.}
Our $\Omega(\kappa)$ lower bound for two-sided matrix--vector algorithms mirrors the $\Omega(\kappa)$ lower bound for the quantum linear systems problem \citep[p.~150502-3]{harrow_quantum_2009}.
In that context, solving implicit sparse systems in $\order(\kappa^{1-\eta} \operatorname{polylog}(n))$ time would imply $\mathsf{BQP} = \mathsf{PSPACE}$ and allow for sub-linear simulation of arbitrary quantum circuits.
While the $\Omega(\kappa)$ scaling in both settings is striking, it remains unclear how to derive one result from the other or reduce them both to some common principle.

\subsection{Notation}

All logarithms are natural.
The vector $\ell_2$ norm and the operator $\ell_2\to\ell_2$ norm are denoted $\norm{\cdot}_2$, and the Frobenius norm is $\norm{\cdot}_{\mathrm{F}}$.
The orthogonal matrices come equipped with a unique rotation-invariant probability distribution, and a draw from this distribution is called a \emph{Haar-random orthogonal matrix}.
The relation \(a \lesssim b\) denotes the relationship  \(a \leq \rC b\) for some universal constant \(\rC > 0\).
When this notation is used in the hypothesis of a result, we take this to mean that $a \le \mathrm{C} b$ for a \warn{sufficiently small} universal constant $\mathrm{C} > 0$.
We write $a \asymp b$ when $a\lesssim b \lesssim a$.
For any function \(f : \bbR \to \bbR\) and symmetric matrix \(\mA\) with eigendecomposition \(\mA = \mU\diag(\lambda_1,\ldots,\lambda_n)\mU^\top\), the matrix function is \(f(\mA)\defeq\mU \diag(f(\lambda_1),\ldots,f(\lambda_n)) \mU^\top\).
 \section{Lower bounds for two-sided algorithms}
\label{sec:two-sided}

In this section, we develop lower bounds for solving linear systems using two-sided matrix--vector algorithms.
Our analysis adopts and extends a relatively new approach to proving such bounds that was developed in recent work of
\ifcolt
\cite{chewi2024query},
\else Chewi, de Dios Pont, Li, Lu, and Narayanan \citep{chewi2024query},
\fi 
who used this technique to prove that $\Omega(\sqrt{\kappa} \log n)$ matvecs are necessary to estimate the trace-inverse \(\tr(\mA^{-1})\) of an SPD matrix \mA to relative error $\norm{\mA^{-1}}_{\mathrm{F}}$.
See also the follow-up work \citep{bakshi_krylov_2023}.

In this section, we present a generalized and refined version of their results that is able to prove lower bounds on the matrix--vector complexity of estimating spectral sums \(\tr(f(\mA))\) or matrix-function--vector products \(f(\mA)\vb\) for arbitrary functions \(f\) and matrices \mA with arbitrary eigenvalues.
The core theorem shows that if a function \(f\) cannot be well approximated by a degree \(t\) polynomial over a domain \(\cS\), then any method that computes \(\tr(f(\mA))\) or \(f(\mA)\vb\) for matrices \mA with eigenvalues belonging to \cS must use \(\Omega(t)\) matvecs.
The proof architecture is adapted from 
\ifcolt
\cite{chewi2024query},
\else 
Chewi et al.\ \cite{chewi2024query},
\fi
but the details have been entirely overhauled to support a general function $f$ and a general set of eigenvalues $\cS$, as well as to obtain large, explicit constants.

\subsection{Matrix--vector complexity and polynomial approximability}

We are interested in computing the solution $\mA^{-1}\vb$ to a linear system of equations.
More generally, we might be interested in computing the matrix-function--vector product $f(\mA) \vb$, for which the linear system problem is the special case $f(x) = \nicefrac{1}{x}$.
It is straightforward to see that if $f$ can be approximated to high accuracy by a degree-$t$ polynomial \warn{on the spectrum of \mA}, then we can compute $f(\mA)\vb$ to high accuracy using a matrix--vector algorithm:
\begin{prop}[Matrix-function-vector product: Upper bound] \label{prop:approx-implies-upper-bound}
    Let $\mA$ be a normal matrix with eigenvalues belonging to a set $\cS\subseteq \bbC$, and suppose that there exists a degree-$t$ polynomial $p(x) = c_0 + c_1 x + \cdots + c_t x^t$ such that
\begin{equation*}
        \max_{\lambda \in \cS} |f(\lambda) - p(\lambda)| \le \eta \max_{\lambda\in\cS} |f(\lambda)|.
    \end{equation*}
Then, the vector \(\tilde\vx \defeq p(\mA)\vb = c_0 \vb + c_1 \mA\vb + \cdots + c_t \mA^t\vb\) can be computed using \(t\) matrix--vector products with \mA, and it approximates \(f(\mA)\vb\) well:
\begin{equation*}
        \norm{\tilde\vx - f(\mA)\vb}_2 \le \norm{p(\mA)-f(\mA)}_2 \norm{\vb}_2 \le \eta \cdot \norm{f(\mA)}_2 \norm{\vb}_2.
    \end{equation*}
\end{prop}

This type of bound is standard.
It shows that we can accurately and efficiently approximate \(f(\mA)\vb\) for any function \(f\) that is close to a polynomial.
In fact, nearly all known matrix--vector algorithms for efficiently computing $f(\mA)\vb$ work directly or indirectly by forming a polynomial approximation to $f$.
Therefore, it is natural to conjecture something of a converse to \cref{prop:approx-implies-upper-bound}: \emph{If a function cannot be approximated by polynomials of small degree, then it is hard for a matrix--vector algorithm to compute $f(\mA)\vb$.}

\subsection {Spectral sums are easier than matrix-function-times-vector}

Rather than proving hardness results directly for computing $f(\mA)\vb$, we instead prove hardness for computing the \emph{spectral sum} \(\tr(f(\mA))\).
The following result shows that if we can approximate \(f(\mA)\vb\) efficiently, then we can also estimate spectral sums efficiently.
As such, lower bounds for the spectral sum problem immediately yield lower bounds for the $f(\mA)\vb$ problem.

\begin{prop}[From matrix-function-times-vector to spectral sum] \label{prop:trace-upper-bound}
    Consider any oracle that takes \mA and \vb as inputs and returns a vector \(\tilde\vx\) such that
\begin{equation*}
        \norm{\tilde\vx - f(\mA)\vb}_2 \le \frac{1}{n} \norm{f(\mA)}_{\mathrm F} \norm{\vb}_2
        \quad
        \text{with prob.\ } \geq \frac56.
    \end{equation*}
Then there exists an algorithm which invokes the oracle once and outputs an estimate $\tilde{\tr}$ such that
\begin{equation*}
        \bigl| \tilde{\tr} - \tr(f(\mA)) \bigr| \le  5\norm{f(\mA)}_{\mathrm F} \quad \text{with prob.\ }\ge \frac{2}{3}.
    \end{equation*}
\end{prop}

\noindent Stochastic trace estimation results of this form are standard; the proof is in \cref{sec:trace-upper-bound}.

\subsection{Lower bound for estimating a spectral sum}

With the necessary preparation in place, we can state a general lower bound for the spectral sum problem in terms of the polynomial (in)approximability of $f$.
We make the following definition:
\begin{definition}[Inapproximable function] \label{def:inapproximable}
    Let \(f:\cS\to\bbR\) be a function on a domain $\cS\subseteq \bbR$.
    The function $f$ is \emph{\((t,\alpha)\)-inapproximable} on \cS if all polynomials \(p\) of degree at most \(t\) satisfy \[\max_{\lambda \in \cS} |p(\lambda) - f(\lambda)| \ge \alpha \cdot \max_{\lambda \in \cS} |f(\lambda)|.\] \end{definition}

Our general result states that, for a sufficiently large symmetric matrix \mA with eigenvalues in the set \cS, at least \(t\) matvecs are needed to estimate spectral sums.
\begin{theorem}[Spectral sum: Lower bound, special case of \cref{thm:black-box-matvec-lower-bound-generic}]
    \label{thm:black-box-matvec-lower-bound}
    Fix a bounded set \(\cS \subseteq \bbR\).
    Let \(f:\cS\to\bbR\) be a function that is \((t,\alpha)\)-inapproximable on \cS.
    Suppose that
    \(
        n \gtrsim  t^6/\alpha^2.
    \)
    Then any matrix--vector algorithm that returns an estimate \(\tilde{\tr}\) such that
    \[
        |\tilde{\tr} - \tr(f(\mA))| \leq 5 \norm{f(\mA)}_{\mathrm F}
        \qquad\qquad
        \text{with prob. }\geq \frac23
    \]
    for all symmetric matrices \(\mA\in\bbR^{n \times n}\) with eigenvalues in \cS must use at least \(\lfloor t/2 \rfloor\) matvecs.
\end{theorem}
\noindent 
Together, \cref{prop:trace-upper-bound,prop:approx-implies-upper-bound} imply that we can estimate $\tr(f(\mA))$ with $\order(t)$ matvecs when the function $f$ admits an accurate degree-$t$ polynomial approximation.
\Cref{thm:black-box-matvec-lower-bound} proves a converse result, showing that if $f$ is \warn{inapproximable} by degree-$t$ polynomials, then at least $\Omega(t)$ matvecs are needed to estimate $\tr(f(\mA))$.
We prove \cref{thm:black-box-matvec-lower-bound} in \cref{app:shyam-generalization} by generalizing and refining the work of
\ifcolt 
\cite{chewi2024query}.
\else 
Chewi \etal \citep{chewi2024query}.
\fi
The basic idea is to construct two matrices $\mA_1$ and $\mA_2$ with a common set of eigenvalues repeated with different multiplicities in such a way that these matrices ``look the same'' to any algorithm using $t/2$ matvecs but for which $|\tr(f(\mA_1)) - \tr(f(\mA_2))| \gg 0$.

\subsection{Lower bounds for trace-inverse}

To use \cref{thm:black-box-matvec-lower-bound} to prove a lower bound on the \(\tr(\mA^{-1})\) and \(\mA^{-1}\vb\) problems, we need to show that \(f(x)=\nicefrac1x\) is inapproximable over an appropriate domain \(\cS\subseteq\bbR\).
Since we are interested in the complexity of these tasks for symmetric matrices \mA, not necessarily SPD, with condition number at most \(\kappa\), we take the domain to be the split interval \(\cS = [-\kappa,-1] \cup [1,\kappa]\).
We show the following inapproximability result:

\begin{lem}[Inapproximability of $\nicefrac1x$ on split interval]
    \label{thm:inverse-inapprox}
    Fix \(\kappa \geq 2\) and an integer $t$.
    Then \(f(x)=\nicefrac1x\) is \((t,\frac12\e^{-2(t+1)/\kappa})\)-inapproximable on \([-\kappa,-1]\cup[1,\kappa]\).
\end{lem}

The proof of \cref{thm:inverse-inapprox} appears in \cref{sec:proving-inapprox}.
Now, we will use use this result to prove lower bounds for computing the trace-inverse.
We have the following result:

\begin{theorem}[Trace-inverse: Lower bound]
    \label{thm:matvec-trace-lower-bound-large-kappa}
    Fix any \(\delta,\eta\in(0,1)\) and suppose that $n\gtrsim \kappa^{(6+\delta)/\eta}$.
    Any matrix--vector algorithm that produces an estimate $\tilde{\tr}$ satisfying
\begin{equation} \label{eq:lower-bound-trace-guarantee}
        \big|\tilde{\tr} - \tr\big(\mA^{-1}\big)\big| \le 5 \norm{\mA^{-1}}_{\mathrm F} \quad \text{with prob.\ } \ge \frac{2}{3}
    \end{equation}
for all symmetric matrices $\mA \in \bbR^{n\times n}$ with condition number $\kappa$
    requires at least $\frac{1-\eta}4 \kappa \log n$ matvecs.
\end{theorem}

\begin{proof}
    Fix any \(t \leq \frac{1-\eta}{4}\kappa\log(n) - 1\).
    Then, we know that \(f(x)=\nicefrac1x\) is \((t,\alpha)\)-inapproximable on \([-\kappa,-1]\cup[1,\kappa]\) for
    \[
        \alpha
        = \frac12 e^{-2(t+1)/\kappa}
\geq \frac12 n^{-\frac{1-\eta}{2}},
    \]
    which gives us
    \[
        \frac{t^6}{\alpha^2}
        \lesssim \kappa^6 n^{1-\eta} \log^6 n
        \lesssim n.
    \]
    The last inequality holds because \(n \gtrsim \kappa^{(6+\delta)/\eta}\).
    Appealing to \cref{thm:black-box-matvec-lower-bound} completes the proof.
\end{proof}

It is informative to instantiate this result in a few regimes.
First, taking the limit \(\eta \to 0\) yields:
\begin{corollary}
    For any $\eta > 0$, any algorithm using fewer than \((\frac14 - \eta)\kappa\log n\) two-sided matvecs cannot estimate \(|\tilde\tr - \tr(\mA^{-1})| \leq 5 \norm{\mA^{-1}}_{\mathrm F}\) with probability at least \(\nicefrac23\) for all matrices \mA.
\end{corollary}
This result gives us the sharpest constant in the complexity, but it does not specify the dimension $n$ needed to produce a hard instance.
If we take $\eta \to 1$, we see that it takes $\Omega(\kappa \log n)$ matvecs to estimate the spectral sum when \(n \gtrsim \kappa^{6 + \delta}\) for any $\delta > 0$.
Curiously, our techniques provide no lower bound when \(n \lesssim \kappa^6\).
We leave this as a natural question for future research.

\subsection{Lower bounds for linear systems of equations}

As a corollary of the \cref{thm:matvec-trace-lower-bound-large-kappa}, we show lower bound on the number of matvecs used by any method which achieves a standard guarantee for solving linear systems.

\begin{theorem}[Linear systems: Lower bound] \label{cor:lin-sys-lb}
    Fix any \(\delta,\eta\in(0,1)\) and suppose \(n \gtrsim \kappa^{(6+\delta)/\eta}\).
    Any algorithm using fewer than \(\frac{1-\eta}4\kappa\log n\) two-sided matvecs with a matrix \mA cannot return a vector \(\tilde\vx\) such that
\begin{equation} \label{eq:lin-sys-lb-1}
        \norm{\smash{\mA\tilde\vx - \vb}}_2 \leq \frac1n \norm{\vb}_2 \quad \text{with prob.\ } \ge \frac{5}{6}
    \end{equation}
for all symmetric matrices $\mA \in \bbR^{n\times n}$ with \(\cond(\mA)\leq\kappa\) and all vectors \(\vb\in\bbR^n\).
    Moreover, attaining the guarantee
\begin{equation} \label{eq:lin-sys-lb-2}
        \norm{\smash{\mA\tilde\vx - \vb}}_2 \leq \frac{1}{2} \norm{\vb}_2 \quad \text{with prob.\ } \ge 1 - \frac{1}{10 \log n}
    \end{equation}
for the same class of problems requires $\Omega(\kappa)$ matvecs.
\end{theorem}
The proof of the first result is immediate from the spectral sum lower bound (\cref{thm:matvec-trace-lower-bound-large-kappa}) together with the trace estimation upper bound (\cref{prop:trace-upper-bound}).
The second result requires another ingredient, iterative refinement (\cref{prop:iterative-refinement}).
See \cref{sec:lin-sys-lb} for the full proofs.

An interesting feature of this result is that the hard instance \mA is symmetric, for which there is no distinction between matvecs with \mA and $\mA^\top$.
In particular, our results show that solving linear systems requires $\Omega(\kappa \log(\nicefrac1\eps))$ matvecs \emph{even when \mA is symmetric}.
Of course, solving a linear system $\mA\vx = \vb$ for a general matrix \mA can only be more difficult than the symmetric case, so the lower bound also applies against general matrices \mA with condition number $\kappa$ as well.

A linear system is solved to residual accuracy $\varepsilon$ when
\[
    \norm{{\mA\tilde\vx - \vb}}_2 \leq \eps\norm\vb_2,
\]
which satisfies the bound \cref{eq:lin-sys-lb-1} when \(\eps \asymp \frac1n\).
In this ``high accuracy'' regime, we show that \(t \gtrapprox \frac14\kappa\log(\nicefrac1\eps)\) matvecs are needed, recovering \cref{thm:intro-two-sided-lower-bound}.
Further, in the ``low accuracy'' regime where \(\eps \leq \frac12\), we show that \(t \gtrsim \kappa\) matvecs are needed.

\subsection{Inapproximability of the inverse function on a split interval}
\label{sec:proving-inapprox}

All that remains to prove \cref{thm:intro-two-sided-lower-bound} is the polynomial inapproximability result (\cref{thm:inverse-inapprox}), which states that the function \(f(x) = \nicefrac1x\) requires a degree \(t = \Omega(\kappa\log(\nicefrac1\eps))\) polynomial to approximate it to accuracy \(\eps\).
The proof starts with a standard bound on the inapproximability of the function \(f(x)=1/x\) on a non-split interval.

\begin{importedlemma}[Inapproximability of $\nicefrac{1}{x}$ on interval; \protect{
    \ifcolt \citet[Cor.~2.2]{kraus2010polynomial}
    \else Kraus \etal \cite[Cor.~2.2]{kraus2010polynomial}
    \fi}]
    \label{implem:non-split-inapprox}
    The function \(f(x)=1/x\) is \((t,\alpha)\)-inapproximable on the interval \([1,b]\) for
    \[
        \alpha
        = \frac12 \bigl(1-\tfrac{2}{\sqrt b+1}\bigr)^{t+1}\bigl(1-\tfrac{1}{\sqrt b}\bigr)^2
        \ge \frac12 e^{-4(t+1)/\sqrt{b}}.
    \]
\end{importedlemma}

Using the fact that \(f(x)=1/x\) is odd, we reduce the inapproximability of \(f\) on the split interval \([-\kappa, -1] \cup [1, \kappa]\) to the inapproximability of \(f\) on the one-sided interval \([1,\kappa^2]\), yielding \cref{thm:inverse-inapprox}.

\begin{proof}[Proof of \cref{thm:inverse-inapprox}]
    Let \(p\) be any polynomial of degree at most \(t\), which we think of as an approximation to $f(x) = \nicefrac{1}{x}$ on the split interval $\cS \defeq [-\kappa,-1]\cup [1,\kappa]$.
Extract the odd part
    \(
        \tilde p(x) \defeq \frac12(p(x) - p(-x))
    \) of $p$,
    which is a polynomial of degree at most \(t\).
    Note that
    \[
        \tilde p(x) - f(x)
        = \frac{p(x) - f(x)}2 - \frac{p(-x) + f(x)}2.
    \]
    Since \(f(x) = -f(-x)\) and \cS is symmetric,
    \begin{align*}
        |\tilde p(x) - f(x)|
        \leq \frac{|p(x) - f(x)|}2 + \frac{|p(-x) - f(-x)|}2
        \leq \max_{x\in\cS}|p(x) - f(x)|.
    \end{align*}
    That is, \(\tilde p\) is no worse than \(p\) at approximating \(f\).
    Next, since \(\tilde p(x)\) is an odd polynomial, we can write \(\tilde p(x) = x q(x^2)\) for some polynomial \(q\) of degree at most \(t_q \defeq \frac12(t-1)\).
    Then, for any \(x \in \cS\),
    \[
        |\tilde p(x) - f(x)|
        = |xq(x^2) - 1/x|
        = x |q(x^2) - 1/x^2|
        \geq |q(x^2) - 1/x^2|.
    \]
    Defining \(y \defeq x^2 \in [1,\kappa^2]\) then \cref{implem:non-split-inapprox} implies that
    \[
        \max_{x\in\cS}|p(x) - f(x)| \ge \max_{x \in \cS} |\tilde p(x) - f(x)|
        \geq \max_{y \in [1,\kappa^2]} |q(y) - 1/y|
        \geq \frac12 e^{-4(t_q+1)/\kappa}.
    \]
    Substituting \(t_q = \frac12(t-1)\) completes the proof.
\end{proof}

 \section{Lower bound against one-sided algorithms}
\label{sec:one-sided}

In this section, we develop lower bounds for one-sided matrix--vector algorithms.
Our proofs use the hidden random matrix method and require the development of a new \emph{hidden Haar theorem}.

\subsection{Hidden Haar theorem}

One of the most powerful proof techniques for proving lower bounds on matrix--vector algorithms is the \emph{hidden random matrix method} \citep{simchowitz2018tight, BHSW20,amsel_fixed-sparsity_2024}.
The method works by identifying a random matrix \mA with the following property: Up to rotation, \mA contains a submatrix that follows a known distribution and is independent of the matrix--vector queries.
One then shows how the uncertainty induced by the hidden matrix prevents any algorithm from solving a given linear algebra task.
Previous works have taken \mA to be a Wishart matrix \citep{BHSW20,chewi2024query,meyer2024towards,amsel_fixed-sparsity_2024,meyer2026hutchinson} or GOE matrix \citep{simchowitz2018tight,jiang2021optimal,woodruff2022optimal}.

In this work, we use a Haar-random (i.e. uniformly random) orthogonal matrix.
The key ingredient to using this random matrix is a \emph{hidden Haar theorem}, which characterizes the distribution of the unobserved part of the matrix.

\begin{theorem}[Hidden Haar]
    \label{thm:hidden-haar}
    Let \(\mQ \in \bbR^{n \times n}\) be a Haar-random orthogonal matrix, and consider the transcript
\begin{equation*}
        \cT = (\vz_1,\mQ \vz_1,\ldots,\vz_t,\mQ\vz_t)
    \end{equation*}
of a deterministic one-sided matrix--vector algorithm making linearly independent queries $\vz_1,\!\ldots\!,\!\vz_t$.
    Then we can construct an orthogonal matrix \(\mU\in\bbR^{n \times n}\) and a matrix \(\mV\in\bbR^{n \times t}\) with orthonormal columns, depending only on the transcript \cT, such that
\[
        \mQ\mU = \bmat{\mV & \mW} \text{ is a Haar-random orthogonal matrix.}
    \]
Further, conditioned on the transcript \cT, the matrix \(\mW \in \bbR^{n \times (n-t)}\) is distributed uniformly at random over all orthonormal matrices spanning the subspace orthogonal to $\operatorname{range}(\mV)$.
\end{theorem}
\noindent
Here, \mW is the hidden submatrix.
We prove the hidden Haar theorem in \cref{sec:hidden-haar-proof}.
The assumption that the algorithm makes linearly independent queries is without loss of generality, as any algorithm making $t$ matrix--vector queries can be simulated by an algorithm making $\tilde t \le t$ linearly independent queries.
Haar matrices have also been previously used in information-theoretic lower bounds for first-order optimization \cite{carmon2020lower}, though we are unaware of any work that uses them in a matrix--vector context.

\subsection{Lower bounds for linear systems}

We can use the hidden Haar theorem to prove lower bounds on the complexity of one-sided linear system solvers.
We begin by stating some elementary concentration and anticoncentration results for uniformly random points on a sphere.

\begin{prop}[Random vectors on the sphere] \label{prop:random-sphere}
        Let $\vu \in \bbR^n$ be a uniformly random vector on the sphere of radius $R$, and assume $n\gtrsim 1$.
        Then
\begin{enumerate}[label=(\alph*)]
        \item \textbf{One coordinate is not too small.} $|[\vu]_1| \ge 0.2 \cdot R/\sqrt{n}$ with probability $> \nicefrac45$. \label{item:one-not-small}
        \item \textbf{Many coordinates are large.} $[\vu]_1^2 + \cdots + [\vu]_{\lfloor n/2\rfloor}^2  \ge 0.49\cdot R^2$ with probability $>\nicefrac35$.\label{item:many-are-large}
        \item \textbf{Far from any point.} For any vector \vv, $\norm{\vv - \vu}_2^2 \ge 0.99 \cdot R^2$ with probability $>\nicefrac35$.\label{item:usually-far}
    \end{enumerate}
\end{prop}

These results can be directly shown by realizing $\vu = R \cdot \vg / \norm{\vg}_2$ by rescaling a standard Gaussian vector $\vg \sim \cN(\vzero,\Id)$ and applying simple (anti)concentration arguments for Gaussian random variables.
We omit the easy proof.

\begin{theorem}[Lower bound against one-sided solvers] \label{thm:one-sided}
    Consider the linear system of equations $\mQ \vx = \ve_1$, where $\mQ\in\bbR^{n\times n}$ is a Haar-random orthogonal matrix and $n\gtrsim 1$.
    \begin{enumerate}[label=(\alph*)]
        \item \textbf{Low accuracy requires $\Omega(n)$ matvecs.} \label{item:low-accuracy-one-sided}
        Any one-sided matrix--vector algorithm producing a vector $\tilde\vx$ satisfying
\begin{equation*}
            \norm{\mQ\tilde\vx - \ve_1}_2 \le \frac{1}{2} \quad \text{with prob.\ } \ge \frac{2}{3}
        \end{equation*}
requires $\lceil n/2\rceil$ matvecs.
        \item \textbf{Modest accuracy requires $n$ matvecs.} \label{item:high-accuracy-one-sided}
        Any one-sided matrix--vector algorithm producing a vector $\tilde\vx$ satisfying
\begin{equation*}
            \norm{\mQ\tilde\vx - \ve_1}_2 \le \frac{0.2}{\sqrt{n}}  \quad \text{with prob.\ } \ge \frac{2}{3}
        \end{equation*}
requires $n$ matvecs.
    \end{enumerate}
\end{theorem}
\begin{proof}
    By Yao's minimax theorem \cite[Sec.~2.2]{motwani_randomized_1995}, we can assume without loss of generality that the algorithm is deterministic.
    Consider a one-sided matrix--vector algorithm performing $t$ matrix--vector products with transcript \cT, which we can assume without loss of generality are linearly independent.
    Instantiate the matrices \mQ, \mU, \mW, and \mV from the hidden Haar theorem (\cref{thm:hidden-haar}).
    The solution to \(\mQ\vx=\ve_1\) is
\begin{equation*}
        \vx = \mQ^\top \ve_1 = \mU \begin{bmatrix}
            \mV^\top \ve_1 \\ 
            \mW^\top\ve_1
        \end{bmatrix}.
    \end{equation*}
Let $\tilde\vx$ denote the output of the algorithm, and introduce the notation
\begin{equation*}
        \tilde \vx \eqqcolon \mU \begin{bmatrix}
        \tilde \vx_1 \\ \tilde\vx_2
    \end{bmatrix}.
    \end{equation*}
Then
\begin{equation*}
        \norm{\mQ\tilde\vx - \ve_1}_2^2 = \norm{\mV^\top\ve_1 - \tilde\vx_1}_2^2 + \norm{\mW^\top\ve_1 - \tilde\vx_2}_2^2 \ge \norm{\mW^\top\ve_1 - \tilde\vx_2}_2^2.
    \end{equation*}
Therefore, a necessary condition for the solution to satisfy $\norm{\mQ\tilde\vx - \ve_1}_2 \le \varepsilon$ is for the component $\tilde\vx_2$ to satisfy
\begin{equation*}
        \norm{\mW^\top\ve_1 - \tilde\vx_2}_2^2 \le \varepsilon^2.
    \end{equation*}
    
    The matrix $\mV$ is a uniformly random orthonormal matrix, and, conditional on $\cT$, the matrix $\mW$ is uniformly  random among orthonormal matrices spanning $[\operatorname{range}(\mV)]^\perp$.
    Consequently, $\mV^\top\ve_1$ is distributed as the first $t$ entries of a uniformly random vector of the unit sphere and, conditional on the transcript $\mW^\top\ve_1$ is uniformly distributed on the sphere of radius $\sqrt{1 - \norm{\mV^\top \ve_1}_2^2}$.
    
We now prove parts \ref{item:low-accuracy-one-sided} and \ref{item:high-accuracy-one-sided} separately.

\paragraph{Proof of \cref{thm:one-sided}\ref{item:low-accuracy-one-sided}.}
Suppose $t < \lceil n /2 \rceil$.
By \cref{prop:random-sphere}\ref{item:many-are-large},
\begin{equation*}
    1 - \norm{\mV^\top \ve_1}_2^2 \ge 0.49 \quad \text{with prob.\ } > 3/5
\end{equation*}
The component $\tilde\vx_2$ and $\mW$ are independent conditional on \cT and thus, by \cref{prop:random-sphere}\ref{item:usually-far}, 
\begin{equation*}
    \norm{\mW^\top\ve_1 - \tilde\vx_2}_2^2 \ge 0.99(1 - \norm{\mV^\top \ve_1}_2^2) \quad \text{with prob.\ } > 3/5.
\end{equation*}
These events are independent, so they occur with probability $(\nicefrac35)^2 = 0.36 > \nicefrac13$.
Ergo, no one-sided matrix--vector algorithm using $\lceil n/2\rceil$ matvecs can achieve the guarantee
\begin{equation*}
    \norm{\tilde\vx - \vx}_2 \le \frac{1}{2}\norm{\vx}_2 \le \sqrt{0.49\cdot0.99}\cdot \norm{\vx}_2
\end{equation*}
with probability $\ge \nicefrac23$.

\paragraph{Proof of \cref{thm:one-sided}\ref{item:high-accuracy-one-sided}.}
Now suppose $t = n-1$.
By \cref{prop:random-sphere}\ref{item:one-not-small},
\begin{equation*}
    \sqrt{1 - \norm{\mV^\top \ve_1}_2^2} \ge 0.2/\sqrt{n} \quad \text{with prob.\ } > \nicefrac45.
\end{equation*}
Moreover, conditional on $\cT$, the number $\mW^\top\ve_1 \in \bbR^1$ is uniformly random between $\pm (1 - \norm{\mV^\top \ve_1}_2^2)$.
The component $\tilde\vx_2$ can be $<(1-\norm{\mV^\top \ve_1}_2^2)$ with probability at most $\nicefrac12$.
Thus, any matrix--vector algorithm using $n-1$ matvecs must fail to achieve the guarantee
\begin{equation*}
    \norm{\tilde\vx - \vx}_2 \le \frac{0.2}{\sqrt{n}} \norm{\vx}_2
\end{equation*}
probability at least $(\nicefrac12)(\nicefrac45) = \nicefrac25 > \nicefrac13$.
\end{proof}

 \section{Conclusion}

In this paper, we resolve the worst-case matrix-vector complexity of solving linear systems \(\mA\vx=\vb\) for one- and two-sided matrix--vector algorithms, up to small constant factors.
Along the way, we also prove a lower bound on the complexity of estimating the inverse-trace \(\tr(\mA^{-1})\).
Several natural questions remain to be explored in this space.
Most directly, can further work on polynomial inapproximability enable similar matching upper and lower bounds for the \(f(\mA)\vb\) and \(\tr(f(\mA))\) problems for a much wider family of functions \(f\)?

Second, despite the \(\Omega(n)\) worst-case bound on the complexity of solving linear systems with one-sided algorithms, we see that methods like GMRES converge in far fewer than \(n\) iterations for a very large class of matrices.
This is even theoretically proven for certain matrices, like those with (complex) eigenvalues belonging in the right half of the complex plane.
However, it remains unclear: What \emph{natural} and \emph{minimal} assumptions on the structure of \mA characterize the rate at which one-sided matrix--vector methods can solve linear systems?

Lastly, and most technically, we see that the spectral sum lower bound in \cref{thm:black-box-matvec-lower-bound} does not yield a lower bound against estimating the trace to relative error $\varepsilon$.
To estimate $\tr(\mA)$ to relative error $\varepsilon$, the optimal rate is $\Theta(1/\varepsilon)$ \citep{meyer_hutch_2021,meyer2024towards}.
Is it possible to show a $\Omega(t/\varepsilon)$ lower bound for computing $\tr(f(\mA))$ to relative error $\varepsilon$ if $f$ is inapproximable up to degree $t$?

\section*{Acknowledgments}

MD was supported in part by NSF CAREER Grant CCF-233865 and a Google ML and Systems Junior Faculty Award.
ENE acknowledges support by the Miller Institute for Basic Research in Science, University of California Berkeley.
RAM would like to acknowledge the DARPA DIAL and DARPA AIQ programs for providing partial support of this work.
This work was done in part while the authors were visiting the Simons Institute for the Theory of Computing.
We thank Ainesh Bakshi, Chris Cama\~no, Cecilia Chen, Lin Lin, Michael Mahoney, Shyam Narayanan, and John Urschel for helpful conversations.
Special thanks go to Ilse Ipsen for pointing us to the literature on information-based complexity.
 
\appendix
\section{Matrix--vector lower bounds from polynomial inapproximability}
\label{app:shyam-generalization}

In this section, we prove the following theorem:
\begin{theorem}[Spectral sums: Lower bound]
    \label{thm:black-box-matvec-lower-bound-generic}
    Fix a bounded set \(\cS\subset\bbR\).
    Let \(f:\cS\to\bbR\) be a function that is \((t,\alpha)\)-inapproximable on \cS.
    Suppose that
    \[
        n \gtrsim  t^4\cdot \max \big\{ t^2, 1/\alpha \big\}.
    \]
    Then, any matrix--vector algorithm that returns an estimate \(\tilde{\tr}\) such that
    \[
        |\tilde{\tr} - \tr(f(\mA))| \lesssim \frac{\alpha \sqrt n}{t^3}\cdot \norm{f(\mA)}_{\mathrm F}
        \qquad\qquad
        \text{with prob. }\geq \frac23
    \]
    for all symmetric matrices \(\mA\in\bbR^{n \times n}\) with eigenvalues belonging to the set \cS must use at least \(\lfloor t/2 \rfloor\) matvecs.
\end{theorem}
The result \cref{thm:black-box-matvec-lower-bound} stated in the text follows as a direct corollary.

\Cref{thm:black-box-matvec-lower-bound-generic} shows that for any function \(f\) that is hard to approximate using a degree-\(t\) polynomial, estimating \(\tr(f(\mA))\) to additive error \(\norm{f(\mA)}_{\mathrm F}\) requires at least \(\Omega(t)\) matvecs.
The proof is a direct generalization and systematization of the argument of 
\ifcolt 
\citet[Sec.~5]{chewi2024query}.
\else 
Chewi et al. \citep[Sec.~5]{chewi2024query}.
\fi

\begin{remark}[One-sided and two-sided]
    This section concerns algorithms applied to a symmetric matrix \mA, for which there is no distinction between one-sided and two-sided matrix--vector algorithms.
    If a linear algebra problem is hard to solve by a matrix--vector algorithm for a symmetric matrices, then it is also hard to solve for general matrices by two-sided matrix--vector algorithms.
\end{remark}

\subsection{Review: Total variation distance} \label{sec:tv}

We begin with a brief review of the total variation distance and its role in proving computational hardness.
This material is standard, but it may be less familiar to those in numerical analysis.

Given two probability measures $\mu$ and $\nu$ on a common space, the \emph{total variation (TV) distance} \citep[Def.~1.1]{canonne_topics_2022} is the maximum difference in probability between the two measures
\begin{equation*}
    \TV(\mu,\nu) \coloneqq \sup_{\cE} | \mu(\cE) - \nu(\cE) |.
\end{equation*}
As the name suggests, the TV distance is a metric on the set of probability measures.
The TV distance quantifies how easy it is to determine which of two distributions a sample is drawn from:
\begin{importedtheorem}[Neyman--Pearson lemma; \protect{\ifcolt\citet[Lem.~1.4]{canonne_topics_2022}\else Canonne \cite[Lem.~1.4]{canonne_topics_2022}\fi}] \label{impthm:tv-and-distinguishability}
    Let \(\mu\) and \(\nu\) be probability distributions.
    Select one of these two distributions with equal probability, and draws a sample \(x\) from said distribution.
    Any algorithm that observes only \(x\) and guesses which distribution \(x\) was drawn from can succeed with probability at most \(\frac12(1+\TV(\mu,\nu))\).
\end{importedtheorem}

An important and elementary result is the \emph{data processing inequality} \citep[Fact~1.1]{canonne_topics_2022}, which states for any function $f$,
\begin{equation} \label{eq:data-processing}
    \TV(f(x),f(y)) \le \TV(x,y).
\end{equation}
That is, any way of processing the data (i.e., applying a function $f$) can only make it harder it is to distinguish a draw of the random variable $x$ from a draw of the random variable $y$.

A final useful result is the \emph{coupling lemma}.
A joint distribution \((x,y)\) is a coupling of \(\mu\) and \(\nu\) if \(x\) is marginally distributed as \(\mu\) and \(y\) is marginally distributed as \(\nu\).
Total variation exactly measures the maximum probability under which a coupling can have \(x=y\):

\begin{importedlemma}[Couplings lemma; \protect{\ifcolt\citet[Prop.~4.7]{levin2017markov}\else Levin and Peres \cite[Prop.~4.7]{levin2017markov}\fi}]
    \label{implem:tv-coupling}
    Let \(\mu\) and \(\nu\) be probability distributions.
    Then,
    \[
        \TV(\mu,\nu) = \inf\bigl\{\Pr\{x \neq y\} ~:~ (x,y) \text{ is a coupling of } \mu \text{ and } \nu \bigr\}.
    \]
\end{importedlemma}

\subsection{Ingredient 1: Arbitrary matrix--vector algorithms can't beat randomized block Krylov} \label{sec:arbitrary-cant-beat-rbki}

An important class of matrix--vector algorithms are randomized block Krylov algorithms, which generate random vectors and perform matvecs exclusively in order to build (a subset of) the \emph{block Krylov sequence}
\begin{equation} \label{eq:traditonal-block-Krylov}
    \BlockKrylov(\mA;\, b,t) \coloneqq (\mA^j \vg_i : 1\le i \le b, \,0 \le j \le t) \quad \text{for } \vg_1,\ldots,\vg_b \sim \cN(\vzero,\Id).
\end{equation}
The parameter
\(b\) is called the \emph{block size} and \(t\) is called the \emph{depth} of the Krylov subspace.
Randomized block Krylov algorithms\footnote{Note that this definition of randomized block Krylov algorithms is broader than standard usage, including as special cases both randomized \emph{single-vector} Krylov algorithms ($b=1$) and non-adaptive randomized algorithms ($t=1$).} have been shown to be optimal for many tasks in linear algebra \citep{simchowitz2018tight,BHSW20,meyer_hutch_2021,bakshi_krylov_2023}, and they are believed to be optimal for many more \citep{bakshi_low-rank_2022,bakshi_krylov_2023}.
The starting point for this analysis is the following powerful theorem of 
\ifcolt
\cite{chewi2024query}
\else 
Chewi \etal \citep{chewi2024query}, 
\fi
which shows that, when working with a rotationally invariant random matrix, an arbitrary matrix--vector algorithm performing $\ell$ matvecs can be simulated using a randomized block Krylov algorithm with block size \(\ell\) and depth \(\ell+1\).

\begin{importedtheorem}[Reduction to block Krylov; 
\protect{\ifcolt\citet[Lem. 5.16]{chewi2024query}\else Chewi \etal \cite[Lem. 5.16]{chewi2024query}\fi}] \label{impthm:from-general-to-rbk}
    Let \(\mU\in\bbR^{n \times n}\) be a Haar-random orthogonal matrix and $\mD \in \bbR^{n\times n}$ be a diagonal matrix, and set \(\mA = \mU\mD\mU^\top\).
    Suppose that an adaptive deterministic algorithm \cA computes matrix--vector products \(\mA\vx_1,\ldots,\mA\vx_\ell\) where \(\ell^2 < n\).
    Let \(\vg_1,\ldots,\vg_\ell \sim \cN(\vec0,\mI)\), and introduce the \warn{triangular} block Krylov sequence
\begin{equation*}
        \TriBlockKrylov(\mA;\,\ell) \coloneqq (\mA^j\vg_i : i\ge 1,~ j\ge 0,~ i+j \le \ell + 1).
    \end{equation*}
Then there exists a deterministic map \(F\), depending only on the algorithm $\cA$, such that
    \[
        F\big(\TriBlockKrylov(\mA;\,\ell)\,\big)
        \overset{\mathrm{dist}}{=} (\mA\vx_i : 1\le i \le \ell ).
    \]
\end{importedtheorem}

This result shows that, in distribution, any adaptive deterministic sequence of $\ell$ matrix--vector products can be obtained from the \warn{triangular} block Krylov sequence $\TriBlockKrylov(\mA;\,\ell)$, which is a subset of the traditional block Krylov sequence $\BlockKrylov(\mA; \,\ell,\ell+1)$ defined above in \cref{eq:traditonal-block-Krylov}.

\subsection{Ingredient 2: From block Krylov to bilinear power sequences} \label{sec:rbki-to-bilinear-powers}

Our strategy will be to generate two matrices $\mA$ and $\tilde\mA$ for which the triangular block Krylov sequences are close in total variation (TV) distance, but for which $\tr(f(\mA))$ and $\tr(f(\tilde\mA))$ are far apart.
To aid in this task, we introduce another auxiliary object, the \emph{bilinear power sequence}
\begin{equation} \label{eq:bilinear-powers}
    \BilinearPower(\mA;\, b, t) \coloneqq ( \vg_i^\top \mA^j \vg_{i'}^{\vphantom{\top}} : 1\le i,i'\le b \text{ and } 0 \le j \le t) \quad \text{for } \vg_1,\ldots,\vg_b \sim \cN(\vzero,\Id). 
\end{equation}
The following lemma shows that if the bilinear power sequences of diagonal matrices $\mD$ and $\tilde\mD$ are close in TV distance, then the block Krylov sequences of random rotations $\mA = \mU\mD\mU^\top$ and $\tilde\mA = \mU \tilde\mD\mU^\top$ generated by a Haar-random orthogonal matrix \mU will be close as well.

\begin{importedlemma}[Reduction to bilinear powers; similar to \protect{\ifcolt\citet[Lem.~5.8]{chewi2024query}\else Chewi \etal \cite[Lem.~5.8]{chewi2024query}\fi}] \label{implem:from-rbk-to-bilpower}
    Let \(\mD,\tilde\mD\) be diagonal matrices, and fix integers $b,t\in \bbN$.
    In addition, let \mU be a Haar-random orthogonal matrix, and set $\mA \coloneqq \mU \mD\mU^\top, \tilde\mA \coloneqq \mU \tilde\mD\mU^\top$.
Then 
\begin{multline} \label{eq:reduction-to-bilinear-powers-conclusion}
        \TV(\BlockKrylov(\mA;\, b,t),\BlockKrylov(\tilde\mA;\, b,t))
        \ifcolt \\ \fi \leq
        \TV(\BilinearPower(\mD;\, b, 2t),\BilinearPower(\tilde\mD;\, b, 2t)). 
    \end{multline}
\end{importedlemma}
\begin{proof}
    Let $\gamma$ denote the right-hand side of \cref{eq:reduction-to-bilinear-powers-conclusion}.
    By the coupling lemma (\cref{implem:tv-coupling}), there exist a coupling between standard Gaussian vectors $\vg_1,\ldots,\vg_b\sim \cN(\vzero,\Id)$ and $\tilde\vg_1,\ldots,\tilde\vg_b\sim \cN(\vzero,\Id)$ such that,
with probability \(1-\gamma\),
    \begin{align*}
\vg_i^\top \mD^j \vg_{i'} = \tilde \vg_i \tilde\mD^j \tilde\vg_{i'}
        \qquad
        &\text{for all } i,i'\in[b] \text{ and } j \in\{0,\ldots,2t\}.
\intertext{
    This condition implies that
    }
\langle \mD^j\vg_i, \mD^{j'} \vg_{i'} \rangle
        = \langle \tilde\mD^j \tilde\vg_i, \tilde\mD^{j'} \tilde\vg_{i'} \rangle
        \qquad
        &\text{for all } i,i'\in[b] \text{ and } j,j' \in\{0,\ldots,t\}.
\end{align*}
This above condition asserts that the \(\mD^j\vg_i\) vectors are merely a rotation of the \(\tilde\mD^j\tilde\vg_{i}\) vectors.
    That is, for some orthogonal matrix \(\mQ\in\bbR^{n \times n}\), it holds that
    \(
        \mD^j\vg_i = \mQ \cdot \tilde\mD^j\tilde\vg_i
    \) for all $i$ and $j$.
    Next, we define
    \(
        \vz_i \defeq \mU\vg_i
\text{ and }
        \tilde\vz_i \defeq \mU\mQ\vg_i,
    \)
    which are both distributed as \(\cN(\vec0,\mI)\) vectors, and set $\mA \coloneqq \mU \mD\mU^\top$ and $\tilde\mA \coloneqq \mU \mQ \tilde \mD \mQ^\top \mU^\top$ for a Haar-random matrix $\mA$.
    The matrix $\tilde\mA$ has the same distribution as $\mU\tilde\mD\mU^\top$, and it is independent of $\tilde \vz_1,\ldots,\tilde\vz_b$.
    Further, with probability at least $1-\gamma$,
    \[
        \mA^j\vz_i
        = \mU \mD^j \mU^\top \mU \vg_i
        = \mU \mD^j \vg_i
        = \mU \mQ \tilde \mD^j \tilde \vg_i
        = \mU \mQ \tilde \mD^j \mQ^\top \mU^\top \cdot \mU \mQ \tilde\vg_i
        = \tilde\mA^j \tilde\vz_i.
    \]
Therefore, we have manifested a coupling under which \(\BlockKrylov(\mA,b,t)\) and \(\BlockKrylov(\tilde\mA,b,t)\) are equal with probability $\ge 1-\gamma$.
    The comparison \cref{eq:reduction-to-bilinear-powers-conclusion} follows by \cref{implem:tv-coupling}.
\end{proof}

Combining the ingredients to this point, we obtain a general hardness result.

\begin{corollary}[Hardness from bilinear powers; similar to \protect{\ifcolt\citet[Lem.~5.8]{chewi2024query}\else Chewi \etal \cite[Lem.~5.8]{chewi2024query}\fi}]
    \label{cor:diag-dtv-to-matvec-dtv}
    Let \(\mD\) and \(\tilde\mD\) be diagonal matrices, fix parameters $b,t \in \bbN$, and suppose that the bilinear power sequences are close in TV distance:
\begin{equation*}
        \TV(\BilinearPower(\mD; \, b, t), \BilinearPower(\tilde\mD; \, b, t)) \leq \nicefrac16.
    \end{equation*}
Let \mU be a Haar-random orthogonal matrix, and define $\mA \coloneqq \mU\mD\mU^\top$ and $\tilde\mA \coloneqq \mU\tilde\mD\mU^\top$.
    Any matrix--vector algorithm that can distinguish $\mA$ from $\tilde\mA$ with probability at least $\nicefrac{2}{3}$ must compute at least $\min \left\{ b-1, \lfloor t/2\rfloor\right\}$ $\gtrsim \min \{b,t\}$ matvecs.
\end{corollary}

\begin{proof}
    By Yao's minimax principle \citep[Sec.~2.2]{motwani_randomized_1995}, we can restrict our attention to deterministic algorithms.
    Consider any matrix--vector algorithm which can make $\ell$ matvecs, and let $\cT = (\mA\vx_i : 1\le i \le \ell)$ and $\tilde\cT = (\tilde\mA\tilde \vx_i : 1\le i \le \ell)$ denote transcripts of matvecs produced by the algorithm on \mA and $\tilde\mA$, respectively.
    By the data processing inequality \cref{eq:data-processing}, the reduction to triangular block Krylov (\cref{impthm:from-general-to-rbk}), the observation that triangular block Krylov is a subset of a general block Krylov method, and the reduction to bilinear power sequences (\cref{implem:from-rbk-to-bilpower}), it holds that
    \begin{align*}
        \TV(\cT,\tilde\cT)
        &\leq \TV(\TriBlockKrylov(\mA; \, \ell),\TriBlockKrylov(\tilde\mA; \, \ell)) \tag{\cref{impthm:from-general-to-rbk}} \\
        &\leq \TV(\BlockKrylov(\mA; \, \ell+1, \ell),\BlockKrylov(\tilde\mA; \, \ell+1, \ell)) \tag{Data Processing} \\
        &\leq \TV(\BilinearPower(\mD; \, \ell+1,2\ell),\BilinearPower(\tilde\mD; \, \ell+1,2\ell)). \tag{\cref{implem:from-rbk-to-bilpower}}
    \end{align*}
Choose $\ell \coloneqq \min \{ b-1, \lfloor t/2\rfloor\}$.
    By the hypothesis, the right-hand side is bounded by a sufficiently small positive constant.
    Ergo, $\TV(\cT,\tilde\cT) < \nicefrac16$, and we conclude by \cref{impthm:tv-and-distinguishability} that any matrix--vector algorithm for distinguishing between $\mA$ and $\tilde\mA$ using $\ell$ matvecs must succeed with probability less than $\nicefrac{2}{3}$.
\end{proof}

\subsection{Ingredient 3: Constructing a finite set of eigenvalues} \label{sec:eigenvalue-construction}

To complete the proof, all we need is two diagonal matrices $\mD$ and $\tilde\mD$ for which the bilinear power sequences \cref{eq:bilinear-powers} are similar (in TV distance) but where the traces $\tr(f(\mD))$ and $\tr(f(\tilde\mD))$ are far apart.
We need to construct these matrices by only using the fact that \(f\) is inapproximable on some domain \cS.
Since \cS may have infinitely many entries, our first step is to discretize \cS into a much smaller finite set, from which we will extract the eigenvalues of \(\mD\) and \(\tilde\mD\).

\begin{importedtheorem}[Inapproximability on a finite set; \protect{\ifcolt\citet[Thm.~1.14]{rivlin_introduction_1981}\else Rivlin \cite[Thm.~1.14]{rivlin_introduction_1981}\fi}] \label{prop:discretized-domain}
    Let \(f\) be continuous and \((t,\alpha)\)-inapproximable on a bounded set \(\cS\subseteq\bbR\).
    Then there exists values \[\lambda_1,\ldots,\lambda_{t+2}\in\cS\] such that \(f\) is inapproximable on \(\cS' \defeq \{\lambda_1,\ldots,\lambda_{t+2}\}\).
\end{importedtheorem}

Strictly speaking, 
\ifcolt 
\cite{rivlin_introduction_1981}
\else the reference \cite{rivlin_introduction_1981}
\fi 
treats only the case when $\cS$ is an interval or a finite set, but the same proof carries over without issue to general bounded sets $\cS$ provided one uses an appropriately general version of the Chebyshev equioscillation theorem \citep[Thm.~5.1]{devore1993constructive}.

\subsection{Ingredient 4: Constructing eigenvalue multiplicities} \label{sec:eigenvalue-multiplicities}

To complete the proof, all we need is two diagonal matrices \(\mD\) and \(\tilde\mD\) whose bilinear power sequences \cref{eq:bilinear-powers} are similar in TV distance, but whose spectral sums \(\tr(f(\mD))\) and \(\tr(f(\tilde\mD))\) are far apart.
To do so, we shall takes the values \(\lambda_1,\ldots,\lambda_{t+2}\) generated in the last section and construct \(\mD\) and \(\tilde\mD\) by including each number \(\lambda_i\) with an appropriate multiplicity.

We construct these multiplicities in steps.
First, in \cref{implem:poly-inapprox-duality}, we use the inapproximability property to find fractional multiplicities which look similar to polynomial functions, but distinguish the function $f$.
Then, in \cref{implem:poly-inapprox-nonzero}, we modify these multiplicities slightly to ensure other desirable properties, such as ensuring that each number $\lambda_i$ receives some non-negligible mass and that the multiplicities for $\mD$ and $\tilde\mD$ are close in relative distance.
Finally, in \cref{cons:indistinguishable-pair}, we round these fractional multiplicities to integers.
We begin with this program below by constructing fractional multiplicities that distinguish the function $f$.

\begin{lem}[Eigenvalue multiplicities I; variant of \protect{\ifcolt\citet[Prop.~5.5]{chewi2024query}\else Chewi \etal \cite[Prop.~5.5]{chewi2024query}\fi}]
    \label{implem:poly-inapprox-duality}
    Suppose the function \(f\) is \((t,\alpha)\)-inapproximable on \(\{\lambda_1,\ldots,\lambda_{t+2}\}\), and fix \(n \ge 1\).
    Then, there exists non-negative weights \(x_1,\ldots,x_{t+2}\) and \(x_1',\ldots,x_{t+2}'\) such that
    \begin{enumerate}[label=(\alph*)]
        \item \textbf{Mass.} The total mass is \(\sum_{i=1}^{t+2} x_i = \sum_{i=1}^{t+2} \tilde x_i = n\). \label{item:mass}
        \item \textbf{Not distinguished by polynomials.} For $j \in [t]$, it holds that \(\sum_{i=1}^{t+2} x_i \lambda_i^j = \sum_{i=1}^{t+2} x_i \lambda_i^j\). \label{item:indistinguishable-poly}
        \item \textbf{Distinguished by $f$.} It holds that \(\sum_{i=1}^{t+2} x_i f(\lambda_i) - \sum_{i=1}^{t+2} \tilde x_i f(\lambda_i) \geq 2 \alpha n \cdot \max_{i \in [t+2]} |f(\lambda_i)|\).
    \end{enumerate}
\end{lem}
\begin{proof}
    The proof is by linear programming duality.
    Set $L$ to be the maximum value of $\sum_{i=1}^{t+2} x_i f(\lambda_i) - \sum_{i=1}^{t+2} \tilde x_i f(\lambda_i)$ over all $\{x_i\}$ and $\{\tilde x_i\}$ satisfying \ref{item:mass} and \ref{item:indistinguishable-poly}.
    The constraints \ref{item:mass} and \ref{item:indistinguishable-poly} are linear, so $L$ is the optimal value of the following linear program:
    \[
        \max_{\bar\vx \in \bbR^{2(t+2)}} \vc^\top\bar\vx
        \qquad\qquad
        \text{subject to }
        \mM\bar\vx=\vb,~\bar\vx\geq0,
    \]
    where
    \[
        \bar\vx = \bmat{x_1 \\ \vdots \\ x_{t+2} \\ x_1' \\ \vdots \\ x_{t+2}'},
        \:\:
        \vc = \bmat{f(\lambda_1) \\ \vdots \\ f(\lambda_{t+2}) \\ -f(\lambda_1) \\ \vdots \\ -f(\lambda_{t+2})},
        \:\:
        \mM = \bmat{
            1 & \cdots & 1 & 1 & \cdots & 1 \\
            1 & \cdots & 1 & -1 & \cdots & -1 \\
            \lambda_1 & \cdots & \lambda_{t+2} & -\lambda_1 & \cdots & -\lambda_{t+2} \\
            \vdots & \ddots & \vdots & \vdots & \ddots & \vdots \\
            \lambda_1^t & \cdots & \lambda_{t+2}^t & -\lambda_1^t & \cdots & -\lambda_{t+2}^t
        },
        \:\:
        \vb = \bmat{
            2n \\ 0 \\ 0 \\ \vdots \\ 0
        }.
    \]
    Let \(\mV = \big[\lambda_i^j\big]_{i\in[t+2],j\in\{0,\ldots,t\}} \in \bbR^{(t+2)\times(t+2)}\) be the Vandermonde matrix of degree \(t\) defined on the nodes \(\lambda_1,\ldots,\lambda_{t+2}\), and let \(\ve\in\bbR^{t+2}\) be the all-ones vector, so that \(\mM^\top = \sbmat{\ve & \mV \\ \ve & -\mV}\).
    Then, the dual problem is
    \[
        \min_{\vy\in\bbR^{t+2}} \vb^\top\vy
        \qquad\qquad
        \text{subject to }
        \mM^\top\vy \geq \vc.
    \]
We partition \(\vy = \sbmat{w \\ \vy_{\mathrm{bot}}}\) so that \([\mV\vy_{\mathrm{bot}}]_i = p(\lambda_i)\) for some polynomial \(p\) with coefficients determined by \(\vy_{\mathrm{bot}}\).
    The constraint of the dual problem then becomes
    \[
        \mM^\top\vy =
        \bmat{
            \ve & \mV \\ \ve & -\mV
        } \bmat{
            w \\ \vy_{\mathrm{bot}}
        }
        = \bmat{w + p(\lambda_1) \\ \vdots \\ w + p(\lambda_{t+2}) \\ w - p(\lambda_1) \\ \vdots \\ w - p(\lambda_{t+2})}
        \geq
        \bmat{
            f(\lambda_1) \\ \vdots \\ f(\lambda_{t+2}) \\ -f(\lambda_1) \\ \vdots \\ -f(\lambda_{t+2})
        }.
    \]
    Notice that the inequalities above are equivalent to saying that \(|p(\lambda_i) - f(\lambda_i)| \leq w\) for all \(i\).
    Next, since \(\vb^\top\vy = 2n w\), we can rewrite the dual problem as an optimization over polynomials \(p\) and the number \(w\):
    \[
        \min_{\mM^\top\vy\geq\vc} \vb^\top\vy
        = \min_{\substack{\deg(p)\leq t \\ |p(\lambda_i)-f(\lambda_i)|\leq w} } 2nw
        = \min_{\deg(p)\leq t} \max_{i\in[t+2]} |p(\lambda_i) - f(\lambda_i)|
        \geq 2n\alpha \cdot \max_{i \in [t+2]} |f(\lambda_i)|.
    \]
The inequality is inapproximability of the function $f$.
    By weak duality, we conclude that $L\ge 2n\alpha \cdot \max_{i \in [t+2]} |f(\lambda_i)|$, as desired.
\end{proof}

Next, we modify the multiplicities to ensure they are uniformly lower-bounded and close in relative distance.

\begin{corollary}[Eigenvalue multiplicities II; variant of \protect{\ifcolt\citet[Cor.~5.6]{chewi2024query}\else Chewi \etal \cite[Cor.~5.6]{chewi2024query}\fi}]
    \label{implem:poly-inapprox-nonzero}
    Fix a parameter \(\eta \in (0,1/2)\).
    Let \(f\) be \((t,\alpha)\)-inapproximable on \(\{\lambda_1,\ldots,\lambda_{t+2}\}\), and fix \(n\ge 1\).
    Then, there exists non-negative weights \(y_1,\ldots,y_{t+2}\) and \(y_1',\ldots,y_{t+2}'\) such that
    \begin{enumerate}[label=(\alph*)]
        \item \textbf{Mass.} The total mass is \(\sum_{i=1}^{t+2} y_i = \sum_{i=1}^{t+2} \tilde y_i = n\). \label{item:eigenvalue-multiplicities-2-mass}
        \item \textbf{Not distinguished by polynomials.} For \(j\in[t]\), it holds that \(\sum_{i=1}^{t+2} y_i \lambda_i^j = \sum_{i=1}^{t+2} \tilde y_i \lambda_i^j\). \label{item:eigenvalue-multiplicities-2-poly}
        \item \textbf{Distinguished by $f$.} \(\sum_{i=1}^{t+2} y_i f(\lambda_i) - \sum_{i=1}^{t+2} \tilde y_i f(\lambda_i) \geq \eta \alpha n \cdot \max_{i \in [t+2]} |f(\lambda_i)|\).\label{item:distinguished-f-middle}
        \item \textbf{Uniformly lower-bounded.} For every $i \in [t+2]$, it holds that \(y_i \geq \frac{n}{2(t+2)}\) and \(\tilde y_i \geq \frac{n}{2(t+2)}\). \label{item:bounded-from-zero}
        \item \textbf{Close in relative distance.} For every $i \in [t+2]$, it holds that \(|y_i - \tilde y_i| \leq 4\eta \cdot y_i\).\label{item:close-relative-distance}
    \end{enumerate}
\end{corollary}

\begin{proof}
    Let $\{x_i\}$ and $\{\tilde x_i\}$ be the weights furnished by \cref{implem:poly-inapprox-duality}, and set
    \[
        w_i = \tfrac14(x_i + \tfrac{n}{t+2})
        \qquad\qquad
        \text{and}
        \qquad\qquad
        w_i' = \tfrac12(\tilde x_i + \tfrac{n}{t+2})
    \]
and
\[
        y_i = \tfrac{1+\eta}{2}w_i + \tfrac{1-\eta}{2}w_i'
        \qquad\qquad
        \text{and}
        \qquad\qquad
        \tilde y_i = \tfrac{1-\eta}{2}w_i + \tfrac{1+\eta}{2}w_i'.
    \]
The five desired properties are readily verified.
\end{proof}

Using this corollary, we are now ready to define a pair of diagonal matrices which we will use to prove hardness.

\begin{construction}[Indistinguishable pair] \label{cons:indistinguishable-pair}
    Given a function $f$ that is $(t,\alpha)$-inapproximable on $\{\lambda_1,\ldots,\lambda_{t+2}\}$ and an integer $n\in \bbN$, we construct a pair of diagonal matrices $\mD,\tilde \mD \in \bbR^{n\times n}$ that cannot be distinguish from their bilinear power sequences.
Namely, let $\{y_i\}_{i=1}^{t+2}$ and $\{\tilde y_i\}_{i=1}^{t+2}$ be the weights furnished by \cref{implem:poly-inapprox-nonzero}.
    Round each weight to integers $\{N_i\}_{i=1}^{t+2}$ and $\{\tilde N_i\}_{i=1}^{t+2}$ by taking either the floor or the ceiling, while ensuring that $\sum_{i=1}^{t+2} N_i = \sum_{i=1}^{t+2} \tilde N_i = n$.
    Finally, define 
\begin{equation*}
        \mD \coloneqq \diag(\underbrace{\lambda_1,\ldots,\lambda_1}_{\text{$N_1$ times}},\ldots,\underbrace{\lambda_{t+2},\ldots,\lambda_{t+2}}_{\text{$N_{t+2}$ times}}) \quad \text{and} \quad \mD \coloneqq \diag(\underbrace{\lambda_1,\ldots,\lambda_1}_{\text{$\tilde N_1$ times}},\ldots,\underbrace{\lambda_{t+2},\ldots,\lambda_{t+2}}_{\text{$\tilde N_{t+2}$ times}}).
    \end{equation*}
Finally, let $\mU$ be a Haar-random orthogonal matrix, and set $\mA = \mU\mD\mU^\top$ and $\tilde\mA = \mU\tilde\mD\mU^\top$.
\end{construction}

\subsection{Ingredient 5: From bilinear power sequences to Wishart matrices}

Now, we need to bound the total variation distance between the bilinear power sequences sequences \cref{eq:bilinear-powers} for $\mD$ and $\tilde\mD$.
To do so, we use the following formula:

\begin{definition}[Wishart matrix]
    Let $\vg_1,\ldots,\vg_b \sim \cN(\vzero,\Id)$ be independent Gaussian vectors of length $N$, and consider a $b\times b$ matrix \mW with entries $[\mW]_{ii'} \coloneqq \langle \vg_i, \vg_{i'}\rangle$.
    Then the random matrix $\mW$ is said to follow \emph{Wishart distribution} with $N$ \emph{degrees of freedom}.
    We write $\mW\sim\Wishart(b,N)$.
\end{definition}

\begin{prop}[Bilinear power sequences: Wishart formula] \label{prop:bilinear-wishart-formula}
    Let $\mD \in \bbR^{n\times n}$ be a diagonal matrix constructed by repeating numbers $\lambda_1,\ldots,\lambda_{t+2}\in\bbR$ with multiplicities $N_1 + \cdots + N_{t+2} = n$, and let $\vg_1,\ldots,\vg_b \sim \cN(\vzero,\Id)$.
    Then there exist \warn{independent} Wishart matrices 
\begin{equation*}
        \mW^{(1)} \sim \Wishart(b,N_1),\ldots,\mW^{(t+2)} \sim \Wishart(b,N_{t+2})
    \end{equation*}
for which
\begin{equation*}
        \vg_i^\top \mD^q \vg_{i'}^{\vphantom{\top}} = \sum_{j=1}^{t+2} \lambda_j^q \big[\mW^{(j)}\big]_{ii'} \quad \text{for every } q\ge 0 \text{ and } 1\le i,i' \le b.
    \end{equation*}
\end{prop}

\begin{proof}
    Let $\vg_1,\ldots,\vg_b \sim \cN(\vzero,\Id)$ be a Gaussian vectors, and let $\vg_i^{(j)} \in \bbR^{N_j}$ denote the restriction of $\vg_i$ to the dimensions corresponding to where eigenvalue $\lambda_j$ appears in $\mD$.
    Similarly, let $\tilde \vg_i^{(j)}$ denote the restriction of $\vg_i$ to the positions in which eigenvalue $\lambda_j$ appears in $\tilde\mD$.
    Observe that
\begin{equation*}
        \vg_i^\top \mD^q \vg_{i'}^{\vphantom{\top}} = \sum_{j=1}^{t+2} \lambda_j^q \big\langle \vg_i^{(j)},\vg_{i'}^{(j)}\big\rangle \quad \text{for every } q\ge 0 \text{ and } 1\le i,i' \le b.
    \end{equation*}
In particular, we conclude that bilinear power sequences of \mD only depend on the \warn{Gram matrices}
\begin{equation*}
        \mW^{(j)} \coloneqq \Big[ \big\langle \vg_i^{(j)},\vg_{i'}^{(j)}\big\rangle \Big]_{1\le i, i' \le b} \sim \Wishart(b,N_i) \quad \text{for } j=1,\ldots,t+2.
    \end{equation*}
of the collection of restricted Gaussian vectors $\vg_i^{(j)}$, which follow a Wishart distribution.
    The stated conclusion follows.
\end{proof}

This result shows us that the bilinear power sequence for a diagonal matrix $\mD$ is generated by a family of Wishart matrices $\mW^{(1)},\ldots,\mW^{(t+2)}$ with degrees of freedom $N_1,\ldots,N_{t+2}$.
As such, given two diagonal matrices $\mD$ and $\tilde\mD$ with corresponding Wishart matrices $\{\mW^{(j)}\}$ and $\{\tilde\mW^{(j)}\}$, it is tempting to try directly bounding the TV distance between the bilinear power sequences by the TV distance between the families of Wishart matrices, i.e.,
\begin{multline} \label{eq:bilinear-from-wishart-1}
    \TV(\BilinearPower(\mD; \, b,t),\BilinearPower(\tilde\mD; \, b,t)) \ifcolt \\ \fi \le \TV\big(\big(\mW^{(1)},\ldots,\mW^{(t+2)}\big),\, \big(\tilde\mW^{(1)},\ldots,\tilde\mW^{(t+2)}\big)\big).
\end{multline}
But there is a problem: As the families $\{\mW^{(j)}\}$ and $\{\tilde\mW^{(j)}\}$ have different numbers of degrees of freedom, the TV distances between these sequences can be large.
As such, the inequality \cref{eq:bilinear-from-wishart-1} is not powerful enough for our purposes.
Fortunately, a simple variation of the trick works.

\begin{corollary}[From bilinear power sequences to Wishart matrices] \label{cor:from-bilinear-to-wishart}
    Construct diagonal matrices $\mD,\mD' \in \bbR^{n\times n}$ via \cref{cons:indistinguishable-pair}, and define independent Wishart matrices 
\begin{equation*}
        \mW^{(i)} \sim \Wishart(b,N_i) \quad \text{and} \quad \tilde\mW^{(i)} \sim \Wishart(b,\tilde N_i) \quad \text{for } i=1,\ldots,t+2.
    \end{equation*}
and define \(\widehat \mW^{(i)} = \tilde\mW^{(i)} + (y_i - \tilde y_i)\mI\).
    Then,
    \begin{multline}
        \label{eq:bilinear-from-wishart-better}
        \TV(\BilinearPower(\mD; \, b,t),\BilinearPower(\tilde\mD; \, b,t)) 
        \ifcolt \\ \fi \le \TV\bigl((\mW^{(1)},\ldots,\mW^{(t+2)}),\, (\widehat\mW^{(1)},\ldots,\widehat\mW^{(t+2)})\bigr).
    \end{multline}
\end{corollary}

\begin{proof}
    Introduce Gaussian vectors $\vg_1,\ldots,\vg_b \sim \cN(\vzero,\Id)$.
    By \cref{prop:bilinear-wishart-formula}, there exist independent Wishart matrices $\tilde \mW^{(1)} \sim \Wishart(b,\tilde N_1),\ldots,\tilde \mW^{(t+2)}\sim \Wishart(b,\tilde N_{t+2})$ such that
\begin{equation*}
        \vg_i^\top \tilde \mD^q \vg_{i'}^{\vphantom{\top}} = \sum_{j=1}^{t+2} \lambda_j^q \tilde \mW^{(j)}_{ii'} \quad \text{for } q\ge 0 \text{ and } i,i' \in [b].
    \end{equation*}
By properties \ref{item:eigenvalue-multiplicities-2-mass} and \ref{item:eigenvalue-multiplicities-2-poly} of the weights $\{y_i\}$ and $\{\tilde y_i\}$ generated by \cref{implem:poly-inapprox-nonzero}, it holds that
\begin{equation*}
        \sum_{j=1}^{t+2} \lambda_j^q (y_j - \tilde{y}_j) = 0 \quad \text{for } q\in\{0,\ldots,t\}.
    \end{equation*}
Combining the two previous displays, we obtain that
\begin{equation*}
        \vg_i^\top \tilde \mD^q \vg_{i'}^{\vphantom{\top}} = \sum_{j=1}^{t+2} \lambda_j^q [\tilde \mW^{(j)} + (y_j - \tilde y_j)\Id]_{ii'} \quad \text{for } q\in\{0,\ldots,t\} \text{ and } i,i'\in[b].
    \end{equation*}
The conclusion follows by \cref{prop:bilinear-wishart-formula} and the data processing inequality \cref{eq:data-processing}.
\end{proof}

The inequality \cref{eq:bilinear-from-wishart-better} is more useful than \cref{eq:bilinear-from-wishart-1} since the shifted Wishart matrices $\widehat \mW^{(j)}$ to have nearly the same mean as their cousins $\mW^{(j)}$.
In particular, the means
\begin{equation*}
    \E \big[ \widehat \mW^{(j)} \big] = (y_j - \tilde{y}_j + \tilde N_j) \Id \quad \text{and} \quad \E \big[ \mW^{(j)} \big] = N_j \Id
\end{equation*}
differ by at most
\begin{equation} \label{eq:max-mean-difference}
    \big| (y_j - \tilde{y}_j + \tilde N_j) - N_j \big| \le |y_j - N_j| + |\tilde y_j - \tilde N_j| \le 2.
\end{equation}
In the final inequality, we recognize that $N_j$ and $\tilde N_j$ are obtained by rounding $y_j$ and $\tilde y_j$ up or down to the nearest integer.

\subsection{Ingredient 6: Total variation distance between Wishart matrices}

To instantiate \cref{cor:from-bilinear-to-wishart}, we need to bound the total variation distance between two Wishart matrices.
We state a general version of such a result here:

\begin{lem}[Total variation distance between Wishart matrices] \label{lem:tv-dist-wishart}
    Let 
\begin{equation*}
        \mW \sim \Wishart(b,N) \quad \text{and} \quad \tilde \mW \sim \Wishart(b,\tilde N)
    \end{equation*}
be Wishart matrices of common dimension $b$ and degrees of freedom $N$ and $\tilde N$.
    For any $\beta\in\bbR$, it holds that
\begin{equation*}
        \TV(\mW, \tilde \mW + \beta \Id) \lesssim b\cdot \frac{|\tilde N + \beta - N|}{\sqrt{N}} + b^{3/2} \cdot \frac{1}{\sqrt{N}}  + b^2 \cdot \frac{|\tilde N - N|}{N}
    \end{equation*}
\end{lem}

A bound similar to this appears in the proof of \protect{\citep[Lem.~5.7]{chewi2024query}}.
To prove this lemma, we employ the following results:

\begin{importedlemma}[Wishart vs.\ GOE; \protect{\ifcolt\citet[Thm.~1.3]{meyer2026non}\else Meyer \cite[Thm.~1.2]{meyer2026non}\fi}] \label{implem:wishart-vs-goe}
     Let \(\mY \defeq N \mI + \sqrt{2N} \mG\) be a shift and scaling of the GOE matrix\footnote{Recall that an unnormalized GOE matrix $\mG \sim \operatorname{GOE}(b)$ is a random symmetric $b\times b$ matrix where the entries on the upper triangle $([\mG]_{ii'} : i\le i')$ are independent, mean-zero Gaussian random variables.
    The diagonal entries have variance $2$, and the off-diagonal entries have variance $1$.} \(\mG \sim \operatorname{GOE}(b)\), and let \(\mW \sim \Wishart(b,N)\) be Wishart.
    Then, the total variation between \(\mW\) and \(\mY\) is \(\lesssim b^{3/2} / N^{1/2}\).
\end{importedlemma}

\begin{importedlemma}[TV between Gaussians; \protect{
    \ifcolt\citet[Thm.~1.3]{devroye_total_2022}
    \else Devroye \etal \cite[Thm.~1.3]{devroye_total_2022}\fi }] \label{lem:dtv-gaussian-lazy}
    The TV distance between the one-dimensional Gaussian distributions $\cN(\mu_1,\sigma_1^2)$ and $\cN(\mu_2,\sigma_2^2)$ is at most
\[\TV(\cN(\mu_1,\sigma_1^2),\cN(\mu_2,\sigma_2^2)) \lesssim \frac{|\sigma_2^2 - \sigma_1^2|}{\sigma_1^2} + \frac{|\mu_1 - \mu_2|}{\sigma_1} .\]
\end{importedlemma}

With these ingredients in place, we prove \cref{lem:tv-dist-wishart}.

\begin{proof}[Proof of \cref{lem:tv-dist-wishart}]
    Draw shifted GOE matrices
\begin{equation*}
        \mY \sim N \Id + \sqrt{2N} \operatorname{GOE}(b) \quad \text{and} \quad \tilde \mY \sim \tilde N \Id + \sqrt{2\tilde N} \operatorname{GOE}(b).
    \end{equation*}
By \cref{implem:wishart-vs-goe}, the triangle inequality, and the data processing inequality \cref{eq:data-processing},
\begin{equation}  \begin{split} \label{eq:tv-Y}
        \TV(\mW,\tilde\mW + \beta \Id) &\le \TV(\mW, \mY) + \TV(\mY, \tilde\mY + \beta\Id) + \TV(\tilde \mY, \tilde\mW) \\
        &\lesssim \TV(\mY, \tilde\mY + \beta\Id) + \frac{b^{3/2}}{N^{1/2}}.
    \end{split}
    \end{equation}
The upper triangles of the matrices $\mY$ and $\tilde \mY + \beta \Id$ have independent entries, so we may bound the total variation distance by summing over the matrix entries
\begin{multline*}
        \TV(\mY, \tilde\mY + \beta\Id) \le \sum_{1\le i \le i' \le b} \TV([\mY]_{ii'},[\tilde \mY + \beta \Id]_{ii'}) \\
        \lesssim b \TV(\cN(N,4N),\cN(\tilde N + \beta,4\tilde N)) + b^2 \TV(\cN(0,2N),\cN(0,\tilde 2N)).
    \end{multline*}
The first term comes from the $b$ diagonal entries, and the second term comes from the $\binom b2 \lesssim b^2$ entries on the strictly upper triangular portion of the matrix.
    Applying \cref{lem:dtv-gaussian-lazy}, we conclude that
\begin{multline*}
        \TV(\mY, \tilde\mY + \beta\Id) \lesssim b \cdot \left[\frac{|\tilde N - N|}{N} + \frac{|\tilde N + \beta - N|}{\sqrt{N}}\right] + b^2 \cdot \frac{|\tilde N - N|}{N} \\ \asymp b\cdot \frac{|\tilde N + \beta - N|}{\sqrt{N}} + b^2 \cdot \frac{|\tilde N - N|}{N}.
    \end{multline*}
Combining with \cref{eq:tv-Y} yields the stated conclusion.
\end{proof}

\subsection{Putting it all together}

Let us now combine all of the ingredients we have assembled to prove \cref{thm:black-box-matvec-lower-bound}.
We begin with one final lemma.

\begin{lem}[Hard-to-distinguish matrices; variant of \protect{\ifcolt\citet[Lem.~5.7]{chewi2024query}\else Chewi \etal \cite[Lem.~5.7]{chewi2024query}\fi}]
    \label{implem:krylov-diag-dtv}
    Fix parameter \(\eta \in (0,1)\), a function $f$ that is $(t,\alpha)$-inapproximable on $\{\lambda_1,\ldots,\lambda_{t+2}\}$, and a size $n\in\bbN$.
    Let $\mD$ and $\tilde\mD$ by the matrices furnished by \cref{cons:indistinguishable-pair}.
    Then
\begin{enumerate}[label=(\alph*)]
        \item \textbf{Distinguished by $f$.} \(|\tr(f(\mD)) - \tr(f(\tilde\mD))| \geq (\eta \alpha n - 2(t+2))\cdot  \max_{i \in [t+2]} |f(\lambda_i)|\). \label{item:distinguished-f-conclusion}
        \item \textbf{Bilinear power sequences are close.} \label{item:bil-power-close-conclusion}
        It holds that \label{item:bil-power-close}
\begin{equation*}
            \TV(\BilinearPower(\mD; \, b,t),\BilinearPower(\tilde\mD; \, b,t)) \lesssim \frac{b^{3/2} t^{3/2}}{n^{1/2}} + \eta b^2 t.
        \end{equation*}
\end{enumerate}
\end{lem}
\begin{proof}
Property \ref{item:distinguished-f-conclusion} is easily verified from \cref{implem:poly-inapprox-nonzero}\ref{item:distinguished-f-middle}.
    We now verify \ref{item:bil-power-close}. 
    By \cref{cor:from-bilinear-to-wishart}, we have
\begin{multline*}
        \TV \bigl(\BilinearPower(\mD; \, b,t),\BilinearPower(\tilde\mD; \, b,t)\bigr) \\
        \le \TV\bigl((\mW^{(1)},\ldots,\mW^{(t+2)}),\, (\tilde\mW^{(1)} + (y_1 - \tilde{y}_1) \Id,\ldots,\tilde\mW^{(t+2)} + (y_{t+2} - \tilde{y}_{t+2}) \Id)\bigr) \\
        \le \sum_{j=1}^{t+2} \TV\bigl(\mW^{(j)},\tilde\mW^{(j)} + (y_j - \tilde{y}_j) \Id\bigr).
    \end{multline*}
Invoking \cref{lem:tv-dist-wishart}, we deduce the bound
\begin{align*}
        \TV\bigl(&\BilinearPower(\mD; \, b,t),\BilinearPower(\tilde\mD; \, b,t)\bigr) \\
        &\lesssim \sum_{j=1}^{t+2}\left[ b \,\frac{| (y_j - \tilde{y}_j + \tilde N_j) - N_j | + \sqrt{b}}{\sqrt{N}_j} + b^2 \, \frac{|\tilde N_j - N_j|}{N_j}\right].
    \end{align*}
We bound the first term by invoking the maximum mean difference bound \cref{eq:max-mean-difference} to control the numerator and property \cref{implem:poly-inapprox-nonzero}\ref{item:bounded-from-zero} to control the denominator as \(N_j \gtrsim n/t\), obtaining
\begin{equation*}
        \frac{| (y_j - \tilde{y}_j + \tilde N_j) - N_j | + \sqrt{b}}{\sqrt{N}_j} \lesssim \frac{b^{1/2} t^{1/2}}{n^{1/2}}.
    \end{equation*}
To bound the second term, observe that \cref{implem:poly-inapprox-nonzero}\ref{item:bounded-from-zero} and \cref{implem:poly-inapprox-nonzero}\ref{item:close-relative-distance} give us
\begin{equation*}
        \frac{|\tilde N_j - N_j|}{N_j} \le \frac{|\tilde N_j - \tilde y_j| + |\tilde y_j - y_j| + |N_j - y_j|}{y_j - 1} \lesssim \frac{t}{n} + \eta.
    \end{equation*}
Combining this bound with the previous display, we obtain
\begin{equation} \label{eq:nearly-complete-tv-bound}
        \TV(\BilinearPower(\mD; \, b,t),\BilinearPower(\tilde\mD; \, b,t)) \lesssim \frac{b^{3/2} t^{3/2}}{n^{1/2}} + \frac{b^2 t^2}{n} + \eta b^2 t.
    \end{equation}
    
    We end the proof by observing that we can drop the second term in \cref{eq:nearly-complete-tv-bound}.
    There are two cases: Either $n \le bt$ or $n > bt$.
    In the former case, $b^{3/2} t^{3/2} / n^{1/2} \ge bt \ge 1$, and the bound is vacuous since the TV distance is at most 1.
    Therefore, in this case, the second term can be dropped.
    In the latter case, it holds that $b^2 t^2 / n < b^{3/2} t^{3/2} / n^{1/2}$, and the second term is dominated by the first.
\end{proof}

\begin{proof}[Proof of \cref{thm:black-box-matvec-lower-bound-generic}]
    Let \(\{\lambda_1,\ldots,\lambda_{t+2}\}\) be the eigenvalues generated by \cref{prop:discretized-domain}, and instate $\mD$, $\tilde\mD$, $\mA$, and $\tilde\mA$ as generated by \cref{cons:indistinguishable-pair}.
    Set $b \coloneqq t$ and $\eta \coloneqq \mathrm{c} t^{-3}$ for a sufficiently small absolute constant $\mathrm{c} > 0$.
By the hypothesis $n \gtrsim t^6$ and \cref{implem:krylov-diag-dtv}\ref{item:bil-power-close-conclusion},
\begin{equation*}
        \TV\bigl(\BilinearPower(\mD; \, b,t),\BilinearPower(\tilde\mD; \, b,t)\bigr)
        \lesssim \frac{t^3}{n^{0.5}} + \eta t^3
    \end{equation*}
is smaller than the value \(\nicefrac16\) needed to activate \cref{cor:diag-dtv-to-matvec-dtv}.
    Thus, by \cref{cor:diag-dtv-to-matvec-dtv}, at least $\min\{t-1, \lfloor t/2 \rfloor\} = \lfloor t/2\rfloor$ matvecs are needed to distinguish the matrices $\mA$ and $\tilde\mA$ with probability at least \(\nicefrac23\).
    
To finish the proof, we must show that any method able to computing \(\tr(f(\mA))\) to the prescribed accuracy can distinguish \(\mA\) from \(\tilde\mA\).
    First, note that
    \[
        \norm{f(\mA)}_{\mathrm F}^2
        = \sum_{i=1}^{t+2} N_i |f(\lambda_i)|^2 \le \bigg(\sum_{i=1}^{t+2} N_i\bigg) \cdot \max_{i \in [t+2]} |f(\lambda_i)|^2 = n \cdot \max_{i \in [t+2]} |f(\lambda_i)|^2.
    \]
    Combining this with \cref{implem:krylov-diag-dtv}\ref{item:distinguished-f-conclusion}, we obtain
\begin{equation*}
        |\tr(f(\mA)) - \tr(f(\tilde\mA))|
        \ge \left( \mathrm{c} t^{-3} \alpha n - 2(t+2) \right) \max_{i \in [t+2]} |f(\lambda_i)|
        \ge \left( \mathrm{c} t^{-3} \alpha n - 2(t+2) \right) \frac{\norm{f(\mA)}_{\mathrm F}}{\sqrt{n}}.
    \end{equation*}
By the hypothesis $n\gtrsim t^4/\alpha$, the first term dominates the second, and we conclude that
\begin{equation*}
        |\tr(f(\mA)) - \tr(f(\tilde\mA))| \gtrsim \frac{\alpha\sqrt{n}}{t^3} \norm{f(\mA)}_{\mathrm F}.
    \end{equation*}
Therefore, any algorithm that can compute \(\tilde{\tr}\) such that \(|\tilde{\tr} - \tr(f(\mA))| \lesssim \frac{\alpha \sqrt n}{t^3}\norm{\mA}_{\mathrm F}\) can distinguish \(\mA\) from \(\tilde\mA\).
    We know this task requires $\lfloor t/2 \rfloor$ matvecs, completing the proof.
\end{proof}

 \section{Proof of the hidden Haar theorem} \label{sec:hidden-haar-proof}

We let $\bbO(n)$ denote the set of $n\times n$ orthogonal matrices and \(\Haar(\bbO(n))\) denote the Haar distribution over $\bbO(n)$.
The defining property of the Haar distribution is rotational invariance:
\begin{quote}
    If \(\mQ\sim\Haar(\bbO(n))\) and \(\mU\in\bbO(n)\) is deterministic, then \(\mQ\mU \sim \Haar(\bbO(n))\).
\end{quote}
Rotational invariance can extend to cases in which \mU is not deterministic.
Indeed, suppose $\mU \in \bbO(n)$ is \emph{random} but independent of \mQ.
Then, conditioned on the realization of \mU, we observe that \(\mQ\mU \sim \Haar(\bbO(n))\).
Since the conditional distribution of \(\mQ\mU\) does not depend on \mU itself, we infer the following \emph{extended} form of the rotational invariance property:
\begin{fact}[Extended rotational invariance]
    If \(\mQ\sim\Haar(\bbO(n))\) and \(\mU\in\bbO(n)\) is independent of \mQ, then \(\mQ\mU \sim \Haar(\bbO(n))\) and $\mQ\mU$ is independent of $\mU$.
\end{fact}
This rotational invariance can even hold when this orthogonal matrix \mU depends on some of the columns of \mQ, so long as \mU acts only on a subset of coordinates.

\begin{lem}[Reading a few columns is uninformative]
    \label{lem:haar-first-t-cols}
    Fix \(0\le t \le n\), let \(\mQ \sim \Haar(\bbO(n))\) be a Haar-random orthogonal matrix, and suppose \(\mT \in \bbO(n-t)\) depends deterministically on the first \(t\) columns of \mQ.
    Then \(\mQ \sbmat{\mI_t & \mat0 \\ \mat 0 & \mT} \sim \Haar(\bbO(n))\).
\end{lem}

\begin{proof}
    By extended rotational invariance, the Haar-random matrix \mQ can be realized as a product 
\begin{equation*}
        \mQ = \mQ_1 \bmat{\Id_t & \mat0 \\ \mat0 & \mQ_2} \quad \text{for } \mQ_1 \sim \Haar(\bbO(n)) \text{ and } \mQ_2 \sim \Haar(\bbO(n-t)) \text{ independent}.
    \end{equation*}
The matrix $\mQ_2$ is independent of the first $t$ columns of $\mQ$ and thus is independent of \mT.
    So, by extended rotational invariance, $\mQ_2\mT$ is independent of $\mQ_1$.
    We conclude again by extended rotational invariance that 
\begin{equation*}
        \mQ \bmat{\Id_t & \mat0 \\
        \mat0 & \mT} = \mQ_1 \bmat{\Id_t & \mat0 \\
        \mat0 & \mQ_2\mT} \sim \Haar(\bbO(n)).
        \qedhere
    \end{equation*}
\end{proof}

With this result in place, we can now prove the hidden Haar theorem.

\begin{proof}[Proof of \cref{thm:hidden-haar}]
We begin by building the matrices $\mV$ and $\mU$.
    Assume without loss of generality that the queries \(\{\vz_i\}\) are orthonormal.\footnote{Indeed, let $\vz_1,\ldots,\vz_t$ denote a general linearly independent query vectors, and reduce it to an orthonormal set $\tilde\vz_1,\ldots,\tilde\vz_t$ by Gram--Schmidt orthonormalization.
    Every vector $\vz_i$ is a linear combination of the orthonormal query vectors $\{\tilde\vz_j\}$, so the transcript $\cT = (\vz_1,\mQ\vz_1,\ldots,\mQ\vz_t)$ is a deterministic function of $\tilde\cT = (\tilde\vz_1,\mQ\tilde\vz_1,\ldots,\mQ\tilde\vz_t)$.} 
For each $i=1,\ldots,t$, define the query matrix
    \[
        \mZ_i \coloneqq \flatbmat{\vz_1 & \cdots & \vz_i} \in \bbR^{n \times i},
    \]
    and introduce the response matrix 
    \[
        \mV_i \coloneqq  \flatbmat{\vv_1 & \cdots & \vv_i} = \flatbmat{\mQ\vz_1 & \cdots & \mQ\vz_t} \in \bbR^{n \times i}.
    \]
Use any deterministic procedure to complete $\mZ_i$ into a square orthogonal matrix
    \[
        \mU_i = \flatbmat{\mZ_i^{\vphantom{\perp}} & \mZ_i^\perp} \in \bbO(n).
    \]
    Since the algorithm is deterministic, $\mZ_{i+1}$ (and thus $\mU_{i+1}$) are both fully determined by the first $i$ responses $\mV_i$.
Moreover, since the first \(i\) columns of $\mU_{i+1}$ and $\mU_i$ are the equal, there is an orthogonal matrix \(\mT_i \in \bbO(n-i)\) such that
\begin{equation} \label{eq:U-iterative}
        \mU_{i+1} = \mU_i \bmat{\mI_i & \mat0 \\ \mat0 & \mT_i}.
    \end{equation}
Observe that the matrix $\mT_i$ is also fully determined by the responses \(\mV_i\).
We define $\mU \coloneqq \mU_t$ and $\mV \coloneqq \mV_t$ and observe that
\begin{equation*}
        \mQ\mU = \flatbmat{\mV & \mW} \quad \text{for some } \mW \in \bbR^{n\times (n-t)}.
    \end{equation*}
    
    As desired, $\mU$ and $\mV$ depend only on the transcript \cT.
    Next, we must establish that $\mQ\mU = \mQ\mU_t \sim \Haar(\bbO(n))$.
    To do so, we show by induction that $\mQ\mU_i \sim \Haar(\bbO(n))$ for each $i=0,\ldots,t$.
    The base case $i = 0$ is immediate. 
    Next, we assume that $\mQ\mU_i \sim \Haar(\bbO(n))$ and prove $\mQ\mU_{i+1} \sim \Haar(\bbO(n))$.
    By \cref{eq:U-iterative}, it holds that
\begin{equation*}
        \mQ\mU_{i+1} = (\mQ\mU_i)\bmat{\mI_i & \mat0 \\ \mat0 & \mT_i}.
    \end{equation*}
The matrix $\mT_i$ depends only on $\mV_i$, which consists of the first $i$ columns of $\mQ\mU_i$.
    By \cref{lem:haar-first-t-cols}, we conclude that $\mQ \mU_{i+1} \sim \Haar(\bbO(n))$, as desired.
    
    Because the algorithm is deterministic, there is a one-to-one correspondence between transcripts \cT and response matrices \mV.
    Therefore, conditioning on \cT is the same as conditioning on \mV, and we conclude that \mW has the specified conditional distribution.
\end{proof}

 \section{Proof of the fine-grained lower bound} \label{app:fine-grained}

In this section, we establish \cref{c:fine-grained}, which gives a more fine-grained lower bound on the two-sided matrix--vector complexity of linear systems, as parametrized by a number $k$ of outlying singular values and a reduced condition number $\kappa_k = \sigma_{k+1}(\mA) / \sigma_{\mathrm{min}}(\mA)$.
We use \cref{thm:intro-two-sided-lower-bound} together with a reduction from 
\ifcolt\citet[Lem.~34]{derezinski2025fine}
\else Dereziński \etal \cite[Lem.~34]{derezinski2025fine}\fi
. 

We begin by importing the following lower bound from \citet{chewi2024query}:
\begin{importedtheorem}[Trace-inverse: $\Omega(n)$ bound; \protect{\citet[Thm.~4.2]{chewi2024query}}]
    \label{impthm:chewi-n-lower-bound}
    There is a universal constant \(\delta > 0\) such that
    any algorithm that produces an estimate \(\tilde\tr\) satisfying
    \[
        \bigl|\tilde\tr - \tr(\mA^{-1})\bigr| \leq \frac12 \tr(\mA^{-1})
        \qquad
        \text{with prob. } \geq 1-\delta
    \]
    for all SPD matrices \(\mA\in\bbR^{n \times n}\) requires at least \(\Omega(n)\) matvecs.
\end{importedtheorem}

From this lower bound for the trace-inverse, we derive a lower bound for the linear system problem.

\begin{lem}[Linear systems: $\Omega(n)$ lower bound]
    \label{lem:unconditioned-n-lower-bound}
    Any algorithm that takes as input \vb, makes two-sided matrix--vector queries with \mA, and outputs a vector \(\tilde\vx\) such that
    \begin{equation}
        \norm{\mA\tilde\vx - \vb}_2 \leq \frac1{6n} \norm{\vb}_2
        \qquad
        \text{with prob. } \geq \frac56
        \label{eq:fine-grained-error}
    \end{equation}
    for all SPD matrices \(\mA\in\bbR^{n \times n}\) must use at least \(\Omega(n)\) matvecs.
\end{lem}

\begin{proof}
    The proof uses a similar linear-systems-to-trace-inverse reduction to \cref{cor:lin-sys-lb}, but there are additional complications associated with reducing the error to level \(\frac12\) and with boosting the success probability to level $1-\delta$, as specified in the trace-inverse lower bound in \cref{impthm:chewi-n-lower-bound}.
    
    Let \(\delta > 0\) be the universal constant in \cref{impthm:chewi-n-lower-bound}, and let $\vb_1,\ldots,\vb_k$ be iid vectors, each with iid Rademacher entries.
    Then by standard analysis of the Girard--Hutchinson trace estimator \cite[Lem.~2.1]{meyer_hutch_2021}, setting \(k=\cO(\log(1/\delta)) = \cO(1)\) suffices to achieve 
    \begin{equation} \label{eq:trace-inv-gh}
        \left|\frac1k \sum_{i=1}^k \vb_i^\top\mA^{-1}\vb_i - \tr(\mA^{-1})\right| \leq \frac14 \norm{\mA^{-1}}_{\mathrm F}
        \qquad
        \text{with prob. } \geq 1-\frac\delta2.
    \end{equation}

    Suppose there is an algorithm that attains the residual guarantee \cref{eq:fine-grained-error} while using at most at most \(t\) matvecs.
    By taking the geometric median of \(\cO(\log(k/\delta)) = \cO(1)\) independent executions of the algorithm, the success probability of the algorithm can be boosted to $1-\delta/2k$ at the cost of inflating the residual by a factor of at most $1.5$ \citep[Thm.~3.1]{minsker2015geometric}.
    Applying this success-probability boosted procedure to each $\vb_i$ (and noting that \(\norm{\vb_i}_2 = \sqrt n\)), we obtain vectors $\tilde\vx_i$ such that, for each $i$,
\begin{equation*}
        \norm{\mA\tilde\vx_i - \vb_i}_2 \leq \frac1{4\sqrt{n}} 
        \qquad
        \text{with prob. } \geq 1-\frac{\delta}{2k}.
    \end{equation*}
The cost is $\order(tk\log(k/\delta)) = \order(t)$ matvecs.
    By union-bounding over these $k$ outcomes, we see that, with probability at least $1-\nicefrac\delta2$,
\begin{equation*}
        \max_{i\in[k]} \norm{\tilde\vx_i - \mA^{-1}\vb_i}_2 \le \norm{\smash{\mA^{-1}}}_2 \max_{i\in[k]}\norm{\mA\tilde\vx_i - \vb_i}_2 \le \frac1{4\sqrt{n}} \norm{\smash{\mA^{-1}}}_{\mathrm{F}}.
    \end{equation*}
In turn, we find that
\(\tilde\tr \defeq \frac1k \sum_{i=1}^k \vb_i^\top\tilde\vx_i\) satisfies
    \begin{align*}
        \left| \tilde\tr - \tfrac1k\sum_{i=1}^k\vb_i^\top\mA^{-1}\vb_i \right|
        &\leq \frac1k \sum_{i=1}^k \bigl| \vb_i^\top\tilde\vx_i - \vb_i^\top\mA^{-1}\vb_i \bigr| \leq \frac1k \sum_{i=1}^k \norm{\vb_i}_2 \norm{\tilde\vx_i - \mA^{-1}\vb_i}_2 \\
        &\leq \frac1k \sum_{i=1}^k \left(\sqrt{n}\cdot \frac{1}{4\sqrt{n}} \norm{\mA^{-1}}_{\mathrm F}\right) = \frac14 \norm{\mA^{-1}}_{\mathrm F}.
    \end{align*}
Finally, combining with guarantee \cref{eq:trace-inv-gh} and union bounding over each bound succeeding, we see that with probability at least $1-\delta$, it holds that
\[
        \bigl| \tilde\tr - \tr(\mA^{-1}) \bigr|
        \leq \bigl| H_k - \tr(\mA^{-1}) \bigr| + \bigl| \tilde\tr - H_k \bigr|
        \leq \frac12 \norm{\mA^{-1}}_{\mathrm F}
        \leq \frac12 \tr(\mA^{-1}).
    \]
The last inequality is the comparison $\norm{\mM}_{\rm F} \le \tr(\mM)$, which is valid for any positive semidefinite matrix \mM.
    We have computed an estimator satisfying the guarantee in \cref{impthm:chewi-n-lower-bound} using $\order(t)$ matvecs.
    We conclude that $t\ge \Omega(n)$, as desired.
\end{proof}

The proof of \cref{c:fine-grained} is a corollary of the \(\Omega(\kappa\log(1/\eps))\) matvecs lower bound in \cref{thm:intro-two-sided-lower-bound}, the \(\Omega(n)\) lower bound in \cref{lem:unconditioned-n-lower-bound}, and a result from prior work explained below.

\begin{definition}[Complexity]
    Fix any set of square matrices \(\cA \subseteq \bbR^{n \times n}\).
    Then the \emph{matrix--vector complexity} for solving linear systems in \cA, denoted \(T_{\cA}(\eps)\), is the smallest number \(t\) such that there exists a two-sided matrix--vector algorithm which computes at most \(t\) matrix--vector products and returns a vector \(\tilde\vx\) such that
    \[
        \norm{\mA\tilde\vx - \vb}_2 \leq \eps\norm{\vb}_2
        \qquad\text{with prob. } \geq \frac56
    \]
    for all \(\mA\in\cA\) and all \(\vb\in\bbR^n\).
\end{definition}

For instance, if \(\cM(n)\) denotes the set of all \(n \times n\) matrices, then \cref{lem:unconditioned-n-lower-bound} says that \(T_{\cM(n)}(\nicefrac{1}{6n}) = \Omega(n)\).
Similarly, if \(\cB(n,\kappa)\) denotes the set of all \(n \times n\) with condition number at most \(\kappa\), then \(T_{\cB(n,\kappa)}(1/n) = \Omega(\kappa\log(1/n))\) by \cref{cor:lin-sys-lb}.
For the purposes of fine-grained complexity, we define \(\cF(n,k,\kappa_k)\) to be the set of all matrices \mA with \(\sigma_{k+1}(\mA) / \sigma_{\min}(\mA) \leq \kappa_k\).
\citet{derezinski2025fine} shows how we can relate the matvec complexity of the fine-grained problem to that of the coarser problems:

\begin{importedtheorem}[Complexity reduction; \protect{\citet[Lem.~27 and Rem.~8]{derezinski2025fine}}]
    \label{impthm:fine-grained-reduction}
    With \(\cM(n)\), \(\cB(n,\kappa)\), and \(\cF(n,k,\kappa_k)\) as defined above,
    \[
        T_{\cF(n,k,\kappa_k)}(\eps)
        \gtrsim
        T_{\cM(k)}(\eps) + T_{\cB(n-k,\kappa_k)}(\eps).
    \]
\end{importedtheorem}

That is, the worst-case fine-grained matvec complexity of linear system solving is the sum of the matvec complexity of solving an arbitrary \(k \times k\) linear system plus the matvec complexity of solving a \((n-k) \times (n-k)\) linear system with condition number \(\kappa_k\).
We can now prove \cref{c:fine-grained}.

\begin{proof}[Proof of \protect{\cref{c:fine-grained}}]
    Fix an accuracy level $\varepsilon > 0$, and choose $k$ and $n$ sufficiently large so that we can instate \cref{cor:lin-sys-lb} and so that $\nicefrac1{6k} < \varepsilon$ and $\nicefrac1{(n-k)} < \varepsilon$.
    Clearly, the complexity is decreasing as a function of the error tolerance \(\eps\).
    Therefore, \cref{lem:unconditioned-n-lower-bound} implies that $T_{\cM(k)}(\varepsilon) = \Omega(k)$ and \cref{cor:lin-sys-lb} implies that $T_{\cB(n-k,\kappa_k)}(\varepsilon) = \Omega(\kappa_k \log(\nicefrac{1}{\eps})$.
    Invoking \cref{impthm:fine-grained-reduction} yields the desired conclusion:
    \[
        T_{\cF(n,k,\kappa_k)}(\varepsilon) = \Omega(k + \kappa_k\log(1/\eps)). \qedhere 
    \]
\end{proof}

 \section{Deferred proofs}

In this section, we provide additional proofs of statements from the main text.

\subsection{Proof of \cref{prop:trace-upper-bound}} \label{sec:trace-upper-bound}

We let \(\vb\in\bbR^{n}\) be a uniformly random vector such that \(\norm{\vb}_2=\sqrt n\).
    Then, by \citet[Prop.~3.1]{girard_algorithme_1987}\footnote{This article is in French. See, e.g., \citet[Fact.~13.2]{epperly2025make} for a statement in English.}, we know that
    \[
        \E[\vb^\top f(\mA) \vb] = \tr(f(\mA))
        \qquad\text{and}\qquad
        \Var[\vb^\top f(\mA) \vb] \leq 2\norm{f(\mA)}_{\mathrm F}^2.
    \]
    So, with probability at least \(5/6\), we have \(|\vb^\top f(\mA)\vb - \tr(\mA)| \leq \sqrt{12}\norm{f(\mA)}_{\mathrm F}\).
    Then the output \(\tilde\vx\) of our oracle has
    \[
        |\vb^\top f(\mA)\vb - \vb^\top\tilde\vx|
        \leq\norm{\vb}_2 \norm{f(\mA)\vb-\tilde\vx}_2
        \leq \frac1n \norm{\vb}_2^2 \norm{f(\mA)}_{\mathrm F}
        = \norm{f(\mA)}_{\mathrm F}
    \]
    with probability at least \(\nicefrac56\).
    We conclude that the trace estimator \(\tilde\tr = \vb^\top\tilde\vx\) has
    \[
        |\tilde\tr - \tr(f(\mA))|
        \leq |\tr(f(\mA)) - \vb^\top f(\mA)\vb| + |\vb^\top f(\mA)\vb - \vb^\top\tilde\vx|
        \leq (1+\sqrt{12})\norm{f(\mA)}_{\mathrm F} 
        \leq 5\norm{f(\mA)}_{\mathrm F}
    \]
    with probability at least \(\nicefrac23\).

\subsection{Proof of \cref{cor:lin-sys-lb}} \label{sec:lin-sys-lb}

Before proving \cref{cor:lin-sys-lb}, we instate a common technique from the literature.

\begin{prop}[Iterative refinement] \label{prop:iterative-refinement}
    Suppose we are given access to an oracle which, given inputs \mA and \vb, returns a vector \(\tilde\vx\) such that
\begin{equation*}
        \norm{\mA\tilde\vx - \vb}_2 \le \frac{1}{2} \norm{\vb}_2 \quad \text{with prob.\ } 1-\delta. 
    \end{equation*}
Then, there exists an algorithm which uses the oracle $r$ times and performs $r$ one-sided matrix--vector products to produce a solution $\tilde\vx$ satisfying
\begin{equation*}
        \norm{\mA\tilde\vx - \vb}_2 \le 2^{-r} \norm{\vb}_2 \quad \text{with prob.\ } \ge 1-r\delta.
    \end{equation*}
\end{prop}

\noindent This type of result and its proof are standard \cite[Sec.~2.1]{greenbaum_iterative_1997}.
The idea is to use the oracle to produce an approximate solution $\tilde\vx_1 \approx \mA^{-1}\vb$.
Then, we form the residual $\vr_1 = \vb - \mA\tilde\vx_1$, use the oracle to produce an approximate solution $\tilde\vdelta_1 \approx \mA^{-1}\vr_1$, and form the corrected solution $\tilde\vx_2 = \tilde\vx_1 + \vdelta_1$.
Continuing this procedure for $r$ rounds produces a solution $\tilde\vx_r$ satisfying the guarantee; we omit a detailed proof.
Using this result, we prove \cref{cor:lin-sys-lb}.

\begin{proof}[Proof of \cref{cor:lin-sys-lb}]
    Fix any matrix \(\mA\) of condition number at most \(\kappa\) and any vector \vb.
    Let \(\tilde\vx\) be any vector returned by a two-sided matrix--vector algorithm such that
    \[
        \norm{\mA\tilde\vx - \vb}_2 \leq \frac1n \norm{\vb}_2
        \qquad
        \text{with prob. } \geq \frac23.
    \]
    Then, we can ensure that
    \[
        \norm{\tilde\vx - \vx}_2
        = \norm{\mA^{-1}\mA\tilde\vx - \mA^{-1}\vb}_2
        \leq \norm{\mA^{-1}}_2\norm{\mA\tilde\vx-\vb}_2
        \leq \frac1n \norm{\mA^{-1}}_{\mathrm F}\norm{\vb}_2.
    \]
    also holds with probability at least \(\nicefrac23\).
    \cref{prop:trace-upper-bound} tell tells us that this matrix--vector algorithm can estimate \(\tr(\mA^{-1})\) with probability at least \(\nicefrac23\), and \cref{thm:matvec-trace-lower-bound-large-kappa} tells us that our matrix--vector algorithm must have used at least \(\frac14\kappa\log(n)\) matvecs, completing the first part of the proof.

    For the second part of the proof, suppose that a two-sided matrix--vector algorithm uses at most \(t\) matvecs to satisfy \cref{eq:lin-sys-lb-2}.
    We then apply \cref{prop:iterative-refinement} with \(r = \Theta(\log(n))\), resulting in a vector \(\tilde\vx\) that satisfies \cref{eq:lin-sys-lb-1} with a total of \(r(t+1)\) matvecs.
    Since it takes \(\Omega(\kappa\log n)\) matvecs to satisfy \cref{eq:lin-sys-lb-1}, we know that our low-accuracy solver must have used \(t = \Omega(\kappa)\) matvecs.
\end{proof}

 \section{Additional lower bounds}
\label{app:other-metrics}

In this section, we briefly note how our methods give variably tight lower bounds for a variety of other matrix--vector complexities in and around linear system solving.

\subsection{\texorpdfstring{$\ell_2$}{L2} error guarantees}

Using a variant of the proof of \cref{cor:lin-sys-lb}, we obtain lower bounds for the \emph{forward error} measured in the $\ell_2$ norm.

\begin{cor}[Linear systems lower bound in \(\ell_2\) norm] \label{cor:lin-sys-lb-backwards}
    Suppose that \(n \gtrsim \kappa^{31}\).
    Any algorithm using fewer than \(\frac15\kappa\log n\) two-sided matvecs with a matrix \mA cannot return a vector \(\tilde\vx\) such that
\begin{equation*}
        \label{eq:lin-sys-lb-different}
        \norm{\smash{\tilde\vx - \mA^{-1}\vb}}_2 \leq \frac1n \norm{\mA^{-1}\vb}_2 \quad \text{with prob.\ } \ge \frac{5}{6}
    \end{equation*}
for all symmetric matrices $\mA \in \bbR^{n\times n}$ with \(\cond(\mA)\leq\kappa\) and all vectors \(\vb\in\bbR^n\).
    Moreover, attaining the guarantee
\begin{equation*}
\norm{\smash{\tilde\vx - \mA^{-1}\vb}}_2 \leq \frac{1}{2} \norm{\mA^{-1}\vb}_2 \quad \text{with prob.\ } \ge 1 - \frac{1}{10 \log n}
    \end{equation*}
for the same class of problems requires $\Omega(\kappa)$ matvecs. 
\end{cor}

The proof of \cref{cor:lin-sys-lb-backwards} is omitted; the proof is nearly the same as the proof of \cref{cor:lin-sys-lb} found in \cref{sec:lin-sys-lb}.
The key step is realizing the guarantee \cref{eq:lin-sys-lb-different} implies
\begin{equation*}
    \norm{\smash{\tilde\vx - \mA^{-1}\vb}}_2 \le \frac{1}{n} \norm{\smash{\mA^{-1}\vb}}_2 \le \frac{1}{n} \norm{\smash{\mA^{-1}}_2\norm{\vb}}_2 \le \frac{1}{n} \norm{\smash{\mA^{-1}}}_{\mathrm{F}}\norm{\vb}_2,
\end{equation*}
which was the same bound established in the proof of \ref{cor:lin-sys-lb}.
\Cref{cor:lin-sys-lb-backwards} shows that resolving a linear system to relative forward error \(\eps\) also requires \(\Omega(\kappa\log(\nicefrac1\eps))\) matvecs.
The best-known upper bounds for achieving an $\epsilon$-small relative forward error require at most \(\cO(\kappa\log(\kappa/\eps))\) matvecs.
The extra $\log(\kappa)$ factor comes from a comparison between the norms $\norm{\mA^{-1}\vb}_2$ and $\norm{\vb}_2$.
To sharpen the lower bound to match the upper bound, it would suffice to produce a lower bound of \(\kappa\log(\kappa)\) matvecs being needed in the low-accuracy regime \(\eps \asymp 1\).

\subsection{Lower bound for SPD systems}
\label{sec:lin-sys-lower-bound-psd}

Lastly, we show an $\Omega(\sqrt\kappa \log(\nicefrac1\eps))$ lower bound for SPD matrices.
An $\Omega(\sqrt\kappa \log(\nicefrac1\eps))$ lower bound is known for deterministic algorithms \citep[Sec.~7.2]{nemirovskij1983problem} and an $\Omega(\sqrt\kappa / \polylog(\kappa))$ lower bound is known for randomized algorithms \cite{BHSW20}.

\begin{theorem}[Linear systems: Lower bound, PSD Variant]
    Fix any \(\delta,\eta\in(0,1)\) and suppose that \(n \gtrsim \kappa^{(6+\delta)/\eta}\).
    Any randomized algorithm using fewer than \(\frac{1-\eta}4\sqrt\kappa\log n\) two-sided matvecs with a matrix \mA cannot return a vector \(\tilde\vx\) such that
\begin{equation*} 
        \norm{\smash{\mA\tilde\vx - \vb}}_2 \leq \frac1n \norm{\vb}_2 \quad \text{with prob.\ } \ge \frac{5}{6}
    \end{equation*}
for all SPD matrices $\mA \in \bbR^{n\times n}$ with \(\cond(\mA)\leq\kappa\) and all vectors \(\vb\in\bbR^n\).
    Moreover, attaining the guarantee
\begin{equation*}         \norm{\smash{\mA\tilde\vx - \vb}}_2 \leq \frac{1}{2} \norm{\vb}_2 \quad \text{with prob.\ } \ge 1 - \frac{1}{10 \log n}
    \end{equation*}
for the same class of problems requires $\Omega(\sqrt\kappa)$ matvecs.
\end{theorem}
\begin{proof}
    The proof matches that of \cref{cor:lin-sys-lb}, except that it uses \cref{implem:non-split-inapprox} in place of \cref{thm:inverse-inapprox}.
\end{proof}
Lower bounds for the forward error in the \(\ell_2\)- and \(\mA\)-norms are also possible, showing that \(\Omega(\sqrt \kappa\log(1/\eps))\) matvecs are needed in both cases.

\bibliographystyle{halpha}

\begin{thebibliography}{CdDPL{\etalchar{+}}24}
\expandafter\ifx\csname url\endcsname\relax
  \def\url#1{\texttt{#1}}\fi
\expandafter\ifx\csname doi\endcsname\relax
  \def\doi#1{\burlalt{doi:#1}{http://dx.doi.org/#1}}\fi
\expandafter\ifx\csname urlprefix\endcsname\relax\def\urlprefix{}\fi
\expandafter\ifx\csname href\endcsname\relax
  \def\href#1#2{#2}\fi
\expandafter\ifx\csname burlalt\endcsname\relax
  \def\burlalt#1#2{\href{#2}{#1}}\fi

\bibitem[ACK{\etalchar{+}}26]{amsel_fixed-sparsity_2024}
Noah Amsel, Tyler Chen, Feyza~Duman Keles, Diana Halikias, Cameron Musco, and
  Christopher Musco.
\newblock Fixed-sparsity matrix approximation from matrix-vector products.
\newblock {\em SIAM Journal on Matrix Analysis and Applications},
  47(2):483--511, June 2026.
\newblock \doi{10.1137/25M1742710}.

\bibitem[AL86]{axelsson1986rate}
Owe Axelsson and Gunhild Lindskog.
\newblock On the rate of convergence of the preconditioned conjugate gradient
  method.
\newblock {\em Numerische Mathematik}, 48:499--523, 1986.
\newblock \doi{10.1007/BF01389448}.

\bibitem[BCW22]{bakshi_low-rank_2022}
Ainesh Bakshi, Kenneth~L. Clarkson, and David~P. Woodruff.
\newblock Low-rank approximation with $\varepsilon^{1/3}$ matrix-vector
  products.
\newblock In {\em Proceedings of the 54th {Annual} {ACM} {SIGACT} {Symposium}
  on {Theory} of {Computing}}, pages 1130--1143. ACM, June 2022.
\newblock \doi{10.1145/3519935.3519988}.

\bibitem[BHOT24]{boulle_operator_2024}
Nicolas Boullé, Diana Halikias, Samuel~E. Otto, and Alex Townsend.
\newblock Operator learning without the adjoint.
\newblock {\em Journal of Machine Learning Research}, 25(364):1--54, 2024.
\newblock \urlprefix\url{http://jmlr.org/papers/v25/24-0162.html}.

\bibitem[BHSW20]{BHSW20}
Mark Braverman, Elad Hazan, Max Simchowitz, and Blake Woodworth.
\newblock The gradient complexity of linear regression.
\newblock In {\em Conference on {{Learning Theory}}}, pages 627--647. PMLR,
  2020.
\newblock \urlprefix\url{https://proceedings.mlr.press/v125/braverman20a.html}.

\bibitem[BN23]{bakshi_krylov_2023}
Ainesh Bakshi and Shyam Narayanan.
\newblock {Krylov} methods are (nearly) optimal for low-rank approximation.
\newblock In {\em 2023 {IEEE} 64th {Annual} {Symposium} on {Foundations} of
  {Computer} {Science}}, pages 2093--2101. IEEE, 2023.
\newblock \doi{10.1109/FOCS57990.2023.00128}.

\bibitem[BV04]{boyd2004convex}
Stephen Boyd and Lieven Vandenberghe.
\newblock {\em Convex Optimization}.
\newblock Cambridge University Press, 2004.
\newblock \doi{10.1017/CBO9780511804441}.

\bibitem[Can22]{canonne_topics_2022}
Clément~L. Canonne.
\newblock {\em Topics and Techniques in Distribution Testing}.
\newblock Now Publishers, 2022.
\newblock \doi{10.1561/9781638281016}.

\bibitem[CdDPL{\etalchar{+}}24]{chewi2024query}
Sinho Chewi, Jaume de~Dios~Pont, Jerry Li, Chen Lu, and Shyam Narayanan.
\newblock Query lower bounds for log-concave sampling.
\newblock {\em Journal of the ACM}, 71(4):1--42, 2024.
\newblock \doi{10.1145/3673651}.

\bibitem[CDHS20]{carmon2020lower}
Yair Carmon, John~C. Duchi, Oliver Hinder, and Aaron Sidford.
\newblock Lower bounds for finding stationary points {I}.
\newblock {\em Mathematical Programming}, 184(1):71--120, 2020.
\newblock \doi{10.1007/s10107-019-01406-y}.

\bibitem[CW09]{clarkson_numerical_2009}
Kenneth~L. Clarkson and David~P. Woodruff.
\newblock Numerical linear algebra in the streaming model.
\newblock In {\em Proceedings of the forty-first annual {ACM} symposium on
  {Theory} of computing}, {STOC} '09, pages 205--214. Association for Computing
  Machinery, May 2009.
\newblock \doi{10.1145/1536414.1536445}.

\bibitem[CW17]{clarkson_low-rank_2017}
Kenneth~L. Clarkson and David~P. Woodruff.
\newblock Low-rank approximation and regression in input sparsity time.
\newblock {\em Journal of the ACM}, 63(6):1--45, February 2017.
\newblock \doi{10.1145/3019134}.

\bibitem[DL93]{devore1993constructive}
Ronald~A DeVore and George~G Lorentz.
\newblock {\em Constructive approximation}, volume 303.
\newblock Springer Science \& Business Media, 1993.

\bibitem[DLNR25]{derezinski2025fine}
Michal Derezi{\'n}ski, Daniel LeJeune, Deanna Needell, and Elizaveta Rebrova.
\newblock Fine-grained analysis and faster algorithms for iteratively solving
  linear systems.
\newblock {\em Journal of Machine Learning Research}, 26(144):1--49, 2025.
\newblock \urlprefix\url{https://www.jmlr.org/papers/v26/24-1906.html}.

\bibitem[DM24]{derezinski2024recent}
Micha{\l} Derezi{\'n}ski and Michael~W Mahoney.
\newblock Recent and upcoming developments in randomized numerical linear
  algebra for machine learning.
\newblock In {\em Proceedings of the 30th ACM SIGKDD Conference on Knowledge
  Discovery and Data Mining}, pages 6470--6479, 2024.
\newblock \doi{10.1145/3637528.3671461}.

\bibitem[DMR22]{devroye_total_2022}
Luc Devroye, Abbas Mehrabian, and Tommy Reddad.
\newblock The total variation distance between high-dimensional {Gaussians}
  with the same mean.
\newblock {\em arXiv preprint
  \href{http://arxiv.org/abs/1810.08693v7}{arXiv:1810.08693v7}}, February 2022.

\bibitem[DMY25]{derezinski_faster_2025}
Micha\l Derezi\'ski, Christopher Musco, and Jiaming Yang.
\newblock Faster linear systems and matrix norm approximation via multi-level
  sketched preconditioning.
\newblock In {\em Proceedings of the 2025 {Annual} {ACM}-{SIAM} {Symposium} on
  {Discrete} {Algorithms}}, pages 1972--2004. SIAM, January 2025.
\newblock \doi{10.1137/1.9781611978322.62}.

\bibitem[DS26]{DerezinskiSidfordSODA26}
Micha{\l} Derezinski and Aaron Sidford.
\newblock Approaching optimality for solving dense linear systems with low-rank
  structure.
\newblock In {\em Proceedings of the 37th Annual ACM-SIAM Symposium on Discrete
  Algorithms (SODA 2026)}. SIAM / ACM, 2026.
\newblock \doi{10.1137/1.9781611978971.37}.

\bibitem[DTT98]{driscoll_potential_1998}
Tobin~A. Driscoll, Kim-Chuan Toh, and Lloyd~N. Trefethen.
\newblock From potential theory to matrix iterations in six steps.
\newblock {\em SIAM Review}, 40(3):547--578, January 1998.
\newblock \doi{10.1137/S0036144596305582}.

\bibitem[DY24]{derezinski2024solving}
Micha{\l} Derezi{\'n}ski and Jiaming Yang.
\newblock Solving dense linear systems faster than via preconditioning.
\newblock In {\em Proceedings of the 56th Annual ACM Symposium on Theory of
  Computing}, pages 1118--1129, 2024.
\newblock \doi{10.1145/3618260.3649694}.

\bibitem[Epp25]{epperly2025make}
Ethan~N. Epperly.
\newblock {\em Make the most of what you have: Resource-efficient randomized
  algorithms for matrix computations}.
\newblock PhD thesis, California Institute of Technology, 2025.
\newblock \doi{10.7907/pef3-mg80}.

\bibitem[Gir87]{girard_algorithme_1987}
Didier Girard.
\newblock Un algorithme simple et rapide pour la validation croisée
  généralisée sur des problèmes de grande taille, 1987.
\newblock
  \urlprefix\url{https://membres-ljk.imag.fr/Didier.Girard/TR-665-M-IMAG.pdf}.

\bibitem[Gre97]{greenbaum_iterative_1997}
Anne Greenbaum.
\newblock {\em Iterative methods for solving linear systems}.
\newblock Number~17 in Frontiers in applied mathematics. Society for Industrial
  and Applied Mathematics, 1997.
\newblock \doi{10.1137/1.9781611970937}.

\bibitem[Hal25]{halikias_structured_2025}
Diana Halikias.
\newblock {\em Structured Matrix Recovery and Approximation from Matrix-Vector
  Products}.
\newblock {PhD}, Cornell University, 2025.

\bibitem[HHL09]{harrow_quantum_2009}
Aram~W. Harrow, Avinatan Hassidim, and Seth Lloyd.
\newblock Quantum algorithm for linear systems of equations.
\newblock {\em Physical Review Letters}, 103(15):150502, October 2009.
\newblock \doi{10.1103/PhysRevLett.103.150502}.

\bibitem[JPWZ21]{jiang2021optimal}
Shuli Jiang, Hai Pham, David Woodruff, and Richard Zhang.
\newblock Optimal sketching for trace estimation.
\newblock {\em Advances in Neural Information Processing Systems},
  34:23741--23753, 2021.
\newblock \urlprefix\url{https://dl.acm.org/doi/10.5555/3540261.3542079}.

\bibitem[KVZ10]{kraus2010polynomial}
Johannes~K Kraus, Panayot~S Vassilevski, and Ludmil~T Zikatanov.
\newblock Polynomial of best uniform approximation to {$x^{-1}$} and smoothing
  in two-level methods.
\newblock {\em arXiv preprint
  \href{http://arxiv.org/abs/1002.1859v3}{arXiv:1002.1859v3}}, 2010.

\bibitem[LP17]{levin2017markov}
David~A Levin and Yuval Peres.
\newblock {\em {Markov} chains and mixing times}, volume 107.
\newblock American Mathematical Soc., 2017.

\bibitem[MA26]{meyer2026hutchinson}
Raphael~A. Meyer and Haim Avron.
\newblock {Hutchinson's} estimator is bad at {Kronecker-Trace-Estimation}.
\newblock {\em SIAM Journal on Matrix Analysis and Applications}, 47(1):1--32,
  2026.
\newblock \doi{10.1137/23M1595180}.

\bibitem[MDM{\etalchar{+}}23]{randlapack_book}
Riley Murray, James Demmel, Michael~W Mahoney, N~Benjamin Erichson, Maksim
  Melnichenko, Osman~Asif Malik, Laura Grigori, Piotr Luszczek, Micha{\l}
  Derezi{\'n}ski, Miles~E Lopes, Tianyu Liang, Hengrui Luo, and Jack Dongarra.
\newblock Randomized numerical linear algebra: A perspective on the field with
  an eye to software.
\newblock {\em arXiv preprint
  \href{http://arxiv.org/abs/2302.11474v2}{arXiv:2302.11474v2}}, 2023.

\bibitem[Mey24]{meyer2024towards}
Raphael~Arkady Meyer.
\newblock {\em Towards Optimal Matrix-Vector Complexity in Numerical Linear
  Algebra}.
\newblock PhD thesis, New York University Tandon School of Engineering, 2024.
\newblock \urlprefix\url{https://ram900.com/assets/thesis-main_final_v3.pdf}.

\bibitem[Mey26]{meyer2026non}
Raphael~A. Meyer.
\newblock A non-asymptotic bound on the {TV} distance between a {Wishart}
  matrix and an appropriately scaled {GOE} matrix.
\newblock {\em arXiv preprint
  \href{https://arxiv.org/abs/2606.16018}{arXiv:2606.16018}}, 2026.

\bibitem[Min15]{minsker2015geometric}
Stanislav Minsker.
\newblock Geometric median and robust estimation in {Banach} spaces.
\newblock {\em Bernoulli}, pages 2308--2335, 2015.
\newblock \doi{10.3150/14-BEJ645}.

\bibitem[MMMW21]{meyer_hutch_2021}
Raphael~A. Meyer, Cameron Musco, Christopher Musco, and David~P. Woodruff.
\newblock {Hutch++}: Optimal stochastic trace estimation.
\newblock In {\em Proceedings of the 2021 Symposium on Simplicity in
  Algorithms}, pages 142--155. SIAM, 2021.
\newblock \doi{10.1137/1.9781611976496.16}.

\bibitem[MR95]{motwani_randomized_1995}
Rajeev Motwani and Prabhakar Raghavan.
\newblock {\em Randomized Algorithms}.
\newblock Cambridge University Press, 1st edition edition, August 1995.
\newblock \doi{10.1017/CBO9780511814075}.

\bibitem[MRT18]{mohri2018foundations}
Mehryar Mohri, Afshin Rostamizadeh, and Ameet Talwalkar.
\newblock {\em Foundations of Machine Learning}.
\newblock MIT press, 2018.

\bibitem[MT20]{martinsson2020randomized}
Per-Gunnar Martinsson and Joel~A. Tropp.
\newblock Randomized numerical linear algebra: Foundations and algorithms.
\newblock {\em Acta Numerica}, 29:403--572, 2020.
\newblock \doi{10.1017/S0962492920000021}.

\bibitem[MW17]{musco_sublinear_2017}
C.~Musco and D.~P. Woodruff.
\newblock Sublinear time low-rank approximation of positive semidefinite
  matrices.
\newblock In {\em 2017 {IEEE} 58th {Annual} {Symposium} on {Foundations} of
  {Computer} {Science}}, pages 672--683, October 2017.
\newblock \doi{10.1109/FOCS.2017.68}.

\bibitem[NRT92]{nachtigal_how_1992}
Noël~M. Nachtigal, Satish~C. Reddy, and Lloyd~N. Trefethen.
\newblock How fast are nonsymmetric matrix iterations?
\newblock {\em SIAM Journal on Matrix Analysis and Applications},
  13(3):778--795, July 1992.
\newblock \doi{10.1137/0613049}.

\bibitem[NY83]{nemirovskij1983problem}
Arkadij~Semenovi{\v{c}} Nemirovsky and David~Borisovich Yudin.
\newblock {\em Problem Complexity and Method Efficiency in Optimization}.
\newblock Wiley-Interscience, 1983.
\newblock
  \urlprefix\url{https://www2.isye.gatech.edu/~nemirovs/Nemirovskii_Yudin_1983.pdf}.

\bibitem[PS82]{paige_lsqr_1982}
Christopher~C. Paige and Michael~A. Saunders.
\newblock {LSQR}: {An} algorithm for sparse linear equations and sparse least
  squares.
\newblock {\em ACM Transactions on Mathematical Software}, 8(1):43--71, March
  1982.
\newblock \doi{10.1145/355984.355989}.

\bibitem[PV24]{peng_solving_2024}
Richard Peng and Santosh~S. Vempala.
\newblock Solving sparse linear systems faster than matrix multiplication.
\newblock {\em Commun. ACM}, 67(7):79--86, July 2024.
\newblock \doi{10.1145/3615679}.

\bibitem[Riv81]{rivlin_introduction_1981}
Theodore~J. Rivlin.
\newblock {\em An Introduction to the Approximation of Functions}.
\newblock Dover Books on Advanced Mathematics. Dover, 1981.

\bibitem[RT08]{rokhlin_fast_2008}
Vladimir Rokhlin and Mark Tygert.
\newblock A fast randomized algorithm for overdetermined linear least-squares
  regression.
\newblock {\em Proceedings of the National Academy of Sciences},
  105(36):13212--13217, September 2008.
\newblock \doi{10.1073/pnas.0804869105}.

\bibitem[SEAR18]{simchowitz2018tight}
Max Simchowitz, Ahmed El~Alaoui, and Benjamin Recht.
\newblock Tight query complexity lower bounds for {PCA} via finite sample
  deformed {Wigner} law.
\newblock In {\em Proceedings of the 50th Annual ACM SIGACT Symposium on Theory
  of Computing}, pages 1249--1259, 2018.
\newblock \doi{10.1145/3188745.3188796}.

\bibitem[Spi24]{spielman_solve_2024}
Daniel Spielman.
\newblock Solve for x: {Technical} perspective.
\newblock {\em Commun. ACM}, 67(7):78, July 2024.
\newblock \doi{10.1145/3643835}.

\bibitem[SS86]{saad_gmres_1986}
Youcef Saad and Martin~H. Schultz.
\newblock {GMRES}: {A} generalized minimal residual algorithm for solving
  nonsymmetric linear systems.
\newblock {\em SIAM Journal on Scientific and Statistical Computing},
  7(3):856--869, July 1986.
\newblock \doi{10.1137/0907058}.

\bibitem[ST09]{spielman2009smoothed}
Daniel~A. Spielman and Shang-Hua Teng.
\newblock Smoothed analysis: An attempt to explain the behavior of algorithms
  in practice.
\newblock {\em Communications of the ACM}, 52(10):76--84, 2009.
\newblock \doi{10.1145/1562764.1562785}.

\bibitem[TWW88]{traub_information_1988}
J.~F. Traub, G.~W. Wasilkowski, and H.~Wo\'{z}niakowski.
\newblock {\em Information-Based Complexity}.
\newblock Academic Press Professional, Inc., 1988.

\bibitem[WZZ22]{woodruff2022optimal}
David Woodruff, Fred Zhang, and Richard Zhang.
\newblock Optimal query complexities for dynamic trace estimation.
\newblock {\em Advances in Neural Information Processing Systems},
  35:35049--35060, 2022.
\newblock \urlprefix\url{https://dl.acm.org/doi/10.5555/3600270.3602810}.

\bibitem[ZDW13]{zhang2013divide}
Yuchen Zhang, John Duchi, and Martin Wainwright.
\newblock Divide and conquer kernel ridge regression.
\newblock In {\em Conference on Learning Theory}, pages 592--617. PMLR, 2013.
\newblock \urlprefix\url{https://proceedings.mlr.press/v30/Zhang13.html}.

\end{thebibliography}
\newcommand{\etalchar}[1]{$^{#1}$}

\end{document}